\documentclass[12pt]{article}
\usepackage{fullpage}
\usepackage[utf8]{inputenc}
\usepackage{amsthm,amsmath}
\usepackage{hyperref}
\usepackage{mathrsfs}
\usepackage{comment,amssymb}
\usepackage{xcolor}

\title{\bf Promise Constraint Satisfaction and Width\footnote{Work
partially supported by Ministerio de Ciencia e Innovaci\'on (MICIN)
through project PID2019-109137GB-C22 (PROOFS).}}

\author{
Albert Atserias \\
Universitat Polit\`ecnica de Catalunya 
\and 
V\'{\i}ctor Dalmau \\
Universitat Pompeu Fabra}

\newtheorem{lemma}{Lemma}

\newtheorem{corollary}{Corollary}
\newtheorem{theorem}{Theorem}

\newtheoremstyle{mydefexobs}
{6pt}
{6pt}
{}
{}
{\itshape}
{.}
{.5em}
{}

\theoremstyle{mydefexobs}

\newtheorem{example}{Example}
\newtheorem{observation}{Observation}

\newcommand{\commentout}[1]{}

\begin{document}

\maketitle

\thispagestyle{empty}

\newcommand{\bG}{\operatorname{\mathbf{G}}}
\newcommand{\bH}{\operatorname{{\bf H}}}
\newcommand{\bI}{\operatorname{{\bf I}}}
\newcommand{\bJ}{\operatorname{{\bf J}}}
\newcommand{\bA}{\operatorname{{\bf A}}}
\newcommand{\bB}{\operatorname{{\bf B}}}
\newcommand{\bF}{\operatorname{{\bf F}}}
\newcommand{\bK}{\operatorname{{\bf K}}}
\newcommand{\bS}{\operatorname{{\bf S}}}
\newcommand{\bT}{\operatorname{{\bf T}}}
\newcommand{\bv}{\operatorname{{\bf v}}}
\newcommand{\bu}{\operatorname{{\bf u}}}
\newcommand{\ba}{\operatorname{{\bf a}}}
\newcommand{\bb}{\operatorname{{\bf b}}}
\newcommand{\bc}{\operatorname{{\bf c}}}
\newcommand{\bj}{\operatorname{{\bf j}}}
\newcommand{\bs}{\operatorname{{\bf s}}}
\newcommand{\bt}{\operatorname{{\bf t}}}
\newcommand{\bx}{\operatorname{{\bf x}}}
\newcommand{\by}{\operatorname{{\bf y}}}
\newcommand{\dom}{\operatorname{dom}}
\newcommand{\img}{\operatorname{img}}

\newcommand{\invd}{\operatorname{sdr}}
\newcommand{\PCSP}{\operatorname{PCSP}}
\newcommand{\CSP}{\operatorname{CSP}}
\newcommand{\Pol}{\operatorname{Pol}}
\newcommand{\pr}{\operatorname{pr}}
\newcommand{\graph}{\script{G}}
\newcommand{\SA}{\operatorname{SA}}
\newcommand{\LP}{\operatorname{LP}}

\newcommand{\LFP}{\mathrm{LFP}}
\newcommand{\Datalog}{\mathrm{Datalog}}
\newcommand{\INF}{\mathrm{L}_{\infty\omega}}
\newcommand{\FPC}{\mathrm{FPC}}
\newcommand{\CINF}{\mathrm{C}_{\infty\omega}}

\newcommand{\sig}{\sigma}

\newcommand{\script}[1]{\mathscr{#1}}
\newcommand{\EXACTLY}[2]{#1\mathrm{\text{-}\mathrm{IN}\text{-}}#2\mathrm{\text{-}}\mathrm{SAT}}
\newcommand{\NAE}[1]{#1\mathrm{\text{-}\mathrm{NAE}\text{-}\mathrm{SAT}}}
\newcommand{\clique}{\mathbf{K}}

\begin{abstract}
We study the power of the bounded-width consistency algorithm in the
context of the fixed-template Promise Constraint Satisfaction Problem
(PCSP). Our main technical finding is that the template of every PCSP
that is solvable in bounded width satisfies a certain structural
condition implying that its algebraic closure-properties include weak
near unanimity polymorphisms of all large arities. While this
parallels the standard (non-promise) CSP theory, the method of proof
is quite different and applies even to the regime of sublinear
width. We also show that, in contrast with the CSP world, the presence
of weak near unanimity polymorphisms of all large arities does not
guarantee solvability in bounded width. The separating example is even
solvable in the second level of the Sherali-Adams (SA) hierarchy of
linear programming relaxations. This shows that, unlike for CSPs,
linear programming can be stronger than bounded width. A direct
application of these methods also show that the problem
of~$q$-coloring~$p$-colorable graphs is not solvable in bounded or
even sublinear width, for any two constants~$p$ and~$q$ such
that~$3 \leq p \leq q$. Turning to algorithms, we note that
Wigderson's algorithm for~$O(\sqrt{n})$-coloring~$3$-colorable graphs
with~$n$ vertices is implementable in width~$4$. Indeed, by
generalizing the method we see that, for any~$\epsilon > 0$ smaller
than~$1/2$, the optimal width for solving the problem
of~$O(n^\epsilon)$-coloring~$3$-colorable graphs with~$n$ vertices
lies between~$n^{1-3\epsilon}$ and~$n^{1-2\epsilon}$. The upper bound
gives a simple~$2^{\Theta(n^{1-2\epsilon}\log(n))}$-time algorithm
that, asymptotically, beats the
straightforward~$2^{\Theta(n^{1-\epsilon})}$ bound that follows from
partitioning the graph into~$O(n^\epsilon)$ many independent parts
each of size~$O(n^{1-\epsilon})$.
\end{abstract}

\newpage

\setcounter{page}{1}

\section{Introduction} \label{sec:intro}

The input to the constraint satisfaction problem (CSP) is a set of
variables, each ranging over a specified domain of values, as well as
a set of constraints, each binding a finite set of variables to take
values in a relation from a specified set of relations. The problem
asks to find an assignment of values to the variables in such a way
that all the constraints are satisfied. It was first pointed out by
Feder and Vardi \cite{FederVardi1998} that the CSP can be modelled as
the homomorphism problem for relational structures. In this view, the
input to the~CSP is a pair of relational structures, the
\emph{instance}~$\bI$ and the \emph{constraint language}~$\bS$, and we
are asked to find a homomorphism from~$\bI$ to~$\bS$; i.e., a map from
the domain of~$\bI$ to the domain of~$\bS$ in such a way that all the
relations are preserved. In the \emph{fixed-template} variant of the
problem, the constraint language is fixed and part of the definition
of the problem, and the instance is the only input. The fixed-template
CSP with template~$\bS$ is denoted by~$\CSP(\bS)$.

As confirmed in retrospect, one of the most important problems raised
by the seminal work of Feder and Vardi was that of characterizing the
class of fixed-template CSPs of \emph{bounded width}. In short, these
are the fixed-template CSPs whose instances are always correctly
solved by the so-called consistency algorithm with a fixed bound on
its \emph{width}. This class of CSPs was later studied in depth in the
work of Kolaitis and
Vardi~\cite{KolaitisVardi2000a,KolaitisVardi2000b}, where it was shown
to constitute a robust fragment of the class of polynomial-time
solvable~CSPs that admits several equivalent reformulations.

After two decades of work in the area, the question of fully
characterizing the CSPs of bounded-width was eventually answered in
the work of Barto and Kozik~\cite{BartoK14}. Besides resolving one of
the key problems in~\cite{FederVardi1998}, the Barto-Kozik Theorem
produced many important new insights for the theory of CSPs
itself. This line of work eventually led to the
proof of the celebrated Feder-Vardi Dichotomy Conjecture, due to
Zhuk~\cite{Zhuk2020} and, independently,
Bulatov~\cite{Bulatov2017}. This completed the program started by
Feder and Vardi \cite{FederVardi1998}, Jeavons, Cohen, and
Gyssens~\cite{JeavonsCG97}, and Bulatov, Jeavons and
Krokhin~\cite{BulatovJeavonsKrokhin2005}, that aimed to characterize
the class of all polynomial-time solvable fixed-template CSPs by the
algebraic closure properties of their templates. By today, it is fair
to say that almost all important questions raised by the early work of
Feder and Vardi for the fixed-template CSP seem to have been resolved.

Recently, Brakensiek and Guruswami \cite{BrakensiekG18} have put
forward a generalization of the fixed-template CSP as a natural next
step in the development of the theory.  In the \emph{Fixed Template
  Promise Constraint Satisfaction Problem} (PCSP) the template is a
fixed pair~$(\bS,\bT)$ of relational structures such that there is a
homomorphism from~$\bS$ to~$\bT$. The problem has two variants:

  \begin{quote}
    \emph{Search variant of $\PCSP(\bS,\bT)$}: The instance is a relational
    structure~$\bI$ with the promise that there is a homomorphism
    from~$\bI$ to~$\bS$
    and we are asked to find a homomorphism from~$\bI$ to~$\bT$. 
    \smallskip

    \emph{Decision variant of~$\PCSP(\bS,\bT)$}: The instance is a
    relational structure~$\bI$ and we are asked to distinguish the
    case in which there is a homomorphism from~$\bI$ to~$\bS$ from the
    case in which there is not even a homomorphism from~$\bI$
    to~$\bT$. For this variant the promise is that~$\bI$ fullfils one
    of these two conditions.
\end{quote}

To motivate the generalization, it is useful to look at some special
cases. First, it is obvious that~$\PCSP(\bS,\bS)$ is just the same
problem as~$\CSP(\bS)$, hence any fixed-template~CSP is a
fixed-template~PCSP. Second, consider the following graph-coloring
problem. We are given a graph~$\bG$ with the promise that it has an
unknown proper~$p$-coloring for a certain number of colors~$p$, and we
are asked to find a proper~$q$-coloring for a fixed put potentially
larger target number of colors~$q$. By viewing the
proper~$q$-colorings of~$\bG$ as the homomorphisms into the complete
graph~$\bK_q$ with~$q$ vertices,~$\PCSP(\bK_p,\bK_q)$ models the
problem of approximating the chromatic number of a graph.  We refer to
this as the Approximate Graph Coloring Problem with parameters~$p$
and~$q$. The special case with~$p = q$ is exactly the same
as~$\CSP(\bK_q)$, the standard Graph~$q$-Coloring Problem for
undirected graphs, which, for~$q \geq 3$, is one of the twenty-one
NP-complete problems of Karp~\cite{Karp1972}. The problem
of~4-coloring~3-colorable graphs,~$\PCSP(\bK_3,\bK_4)$, was shown
NP-hard in the early~1990's, using the theory of probabilistically
checkable proofs (PCP)~\cite{KhannaLS00}. Not until two decades later
has it been shown that the problem of~5-coloring~3-colorable graphs is
NP-hard, and, more generally, that~$\PCSP(\bK_p,\bK_{2p-1})$ is
NP-hard for any~$p \geq 3$~\cite{BBKO}. The problem of exactly
identifying the pairs~$(p,q)$ with~$p \leq q$ for
which~$\PCSP(\bK_p,\bK_q)$ is NP-hard is one of the leading open
problems in the area, with deep connections with the theories of PCPs
and hardness of approximation~\cite{LundY94}, and the notorious Unique
Games Conjecture~\cite{DinurMR09}.

The theory of Promise PCSP has been further developped (see
\cite{Barto19,BartoBB21,BBKO,BrakensiekG19,BrakensiekGS21,BrakensiekGWZ20,BrandtsWZ20,BrandtsZ21,FicakKOS19,KrokhinO19,WrochnaZ20}
for example) and applied to other problems beyond graph coloring. The
goal of this paper is to initiate a study of the power of the
consistency algorithm for the fixed-template PCSP. Following
\cite{BBKO}, in this new context, the consistency algorithm can be
thought of in the following terms.  Fix a PCSP template~$(\bS,\bT)$
and an integer~$k \geq 1$. We say that~$\PCSP(\bS,\bT)$ is
\emph{solvable in width~$k$} if, for every instance~$\bI$, it holds
that if~$\bI$ is~$k$-consistent with respect to the
left-template~$\bS$, then there is a homomorphism from~$\bI$ to the
right-template~$\bT$. Here, as in the standard consistency algorithm
for the standard CSP, an instance~$\bI$ is said to
be~\emph{$k$-consistent with respect to~$\bS$} if the algorithm that
checks if the set of all subinstances of~$\bI$ with at most~$k$
elements admits a system of compatible homomorphisms into~$\bS$ does
not find a blatant contradiction. For each fixed integer~$k$, this
algorithm runs in polynomial time. When~$\PCSP(\bS,\bT)$ is solvable
in width~$k$ we also say that the template~$(\bS,\bT)$ \emph{has}
width~$k$.

A key insight from the theory of CSPs is that the computational
complexity of a fixed-template CSP is governed by the algebraic
structure of the polymorphisms of its
template~\cite{JeavonsCG97}. This phenomenon led to the so-called
\emph{algebraic approach} to CSPs, of which the aforementioned
Barto-Kozik, Zhuk, and Bulatov Theorems are prime examples. Indeed,
the governing power of polymorphisms is so general that the phenomenon
has been observed to hold for other ways of measuring the complexity
of the problem. In particular, it applies to the analysis of bounded
width~\cite{LaroseZ06}, to the more general settings of descriptive
complexity~\cite{AtseriasBD09} and propositional proof
complexity~\cite{AtseriasO19}, and even to different variants of
the~CSP itself. For~PCSPs, a suitable definition of polymorphisms has
been put forward to show that the computational and width complexity
of a PCSP is, again, governed by the polymorphisms of its
template~\cite{BBKO}.

A natural question at this point is whether the analogue of the
Barto-Kozik Theorem holds for PCSPs. In the language of polymorphisms,
the Barto-Kozik Theorem for CSPs can be stated as the equivalence of
the following two statements:
\begin{enumerate}\itemsep=0pt
\item[(a)] $\CSP(\bS)$ is solvable in some bounded width,
\item[(b)] $\CSP(\bS)$ admits weak near unanimity (WNU) polymorphisms of all large arities.
\end{enumerate}
In definition, this last condition means that, for any large enough
integer~$m$, there exists a homomorphism~$f : \bS^m \to \bS$ from
the~$m$-power~$\bS^m$ into~$\bS$ that satisfies the~\emph{weak near
  unanimity} (WNU) identities:
\begin{equation}
f(y,x,x,\ldots,x) = f(x,y,x,\ldots,x) = \cdots
= f(x,\ldots,x,x,y), \label{eqn:wnu}
\end{equation}
for all choices of~$x$ and~$y$ in the domain of~$\bS$.  

Turning to PCSP, for a fixed template~$(\bS,\bT)$, consider the
following statements:
\begin{enumerate} \itemsep=0pt
\item[(a')] $\PCSP(\bS,\bT)$ is solvable in some bounded width,
\item[(b')] $\PCSP(\bS,\bT)$ admits WNU polymorphisms of all large arities.
\end{enumerate}
Following~\cite{BBKO}, a WNU polymorphism of~$\PCSP(\bS,\bT)$ is a
homomorphism~$f : \bS^m \to \bT$ from the~$m$-power~$\bS^m$ of the
left-template~$\bS$ into the right-template~$\bT$ that satisfies
Equation~\eqref{eqn:wnu} for all~$x$ and~$y$ in the domain of~$\bS$.

The first technically novel result of this paper is that~(b') is a
\emph{necessary} but \emph{not sufficient} condition for (a'):

\begin{center}
\textbf{Main Result}: (a')$\Longrightarrow$ (b') $\not\Longrightarrow$ (a').
\end{center}

To prove that (a')~$\Longrightarrow$ (b') we first establish a
structural result about the template that may be of independent
interest. We show that any template~$(\bS,\bT)$ whose PCSP is solvable
in some bounded width satisfies the following property: either a
proper subset of~$S\times S$ can be obtained by composing two binary
projections of the relations in~$\bS$, where~$S$ is the domain
of~$\bS$, or else the problem is trivial because~$\bT$ contains a
reflexive tuple and then \emph{any} instance admits a homomorphism
to~$\bT$.  In its proof we need to resort to a probabilistic
construction of a large instance~$\bI$ that is \emph{sparse enough} to
be~$k$-consistent with respect to~$\bS$, for any fixed~$k$, but
\emph{dense enough} to not admit a homomomorphisms into~$\bT$.

The reasoning that goes into the analysis of the probabilistic
construction of this instance is reminiscent of the methods for
proving lower bounds for Resolution in propositional proof complexity,
going back to the influential work of Chv\'atal and
Szem\'eredi~\cite{ChvatalSzemeredi1988} and posterior follow-ups
(particularly,~\cite{BeameCMM05}
and~\cite{MolloySalavatipour2007}). Indeed, as in these related works,
our analysis shows that the random instance~$\bI$ is~$k$-consistent
with respect to~$\bS$ for~$k$ as large as~$\epsilon n$, where~$n$ is
the number of elements in~$\bI$, and~$\epsilon$ is a fixed positive
constant that depends only on~$\bS$ and~$\bT$. This means that the
necessary condition (b') applies also to all PCSPs that are solvable
in \emph{sublinear width}; i.e., in width~$k = k(n) = o(n)$, where~$n$
is the number of elements in the instance. It should be pointed out
that, for standard CSPs, it was already known that the Barto-Kozik
Theorem can also be strengthened to show that a fixed-template~CSP is
solvable in bounded width if and only if it is solvable in sublinear
width (see, e.g.,~\cite{AtseriasO19}).

To prove that (b')~$\not\Longrightarrow$ (a'), we analyze the
polymorphisms of a specific \emph{Boolean} PCSP template, i.e., one
with a two-element domain~$\{0,1\}$,
namely~$(\EXACTLY{2}{4},\NAE{4})$. We first show that this template
admits WNU polymorphisms of all large arities, and then apply the
structural result of the previous paragraph to conclude that it is not
solvable in bounded width. This analysis also led us to conclude that,
for any two integers~$s$ and~$r$ such that~$0 < s < r$ and~$r > 2$,
the Boolean~$\PCSP(\EXACTLY{s}{r},\NAE{r})$ is not solvable in bounded
width, but is solvable in the second level of the Sherali-Adams
hierarchy applied to its basic linear programming relaxation. This is
in sharp constrast with the status for standard~CSPs for which, as is
known, the fixed-template CSPs that are solvable in bounded-width and
those that are solvable in some fixed-level of the Sherali-Adams
hierarchy coincide (this follows, e.g., from the Barto-Kozik Theorem
combined with the results in~\cite{AtseriasBD09}
and~\cite{AtseriasM13}).  To show that~$\PCSP(\EXACTLY{s}{r},\NAE{r})$
can be solved in the second level of the Sherali-Adams hierarchy we
build on the recent results on Boolean PCSPs from
\cite{BrakensiekG18}.

As a corollary to the aforementioned structural result we obtain a
complete classification of the Approximate Graph Coloring Problems
that are solvable in bounded width. We show that~$\PCSP(\bK_p,\bK_q)$
is not solvable in bounded width for \emph{any} two constants~$p$
and~$q$, such that~$p \leq q$, unless~$p = 1$ or~$p = 2$.

\begin{corollary} For any two integers~$p$ and~$q$ such
  that~$1 \leq p \leq q$, the following statements are equivalent:
\begin{enumerate} \itemsep=0pt
\item[(a)] $\PCSP(\bK_p,\bK_q)$ is solvable by the consistency algorithm in width~$3$,
\item[(b)] $\PCSP(\bK_p,\bK_q)$ is solvable by the consistency algorithm in bounded width,
\item[(c)] $\PCSP(\bK_p,\bK_q)$ is solvable by the consistency algorithm in sublinear width,
\item[(d)] $p = 1$ or $p = 2$.
\end{enumerate} 
\end{corollary}

\noindent We note that this classification agrees with the one
predicted by the NP-hardness results that would follow from certain
variants of the Unique Games Conjecture \cite{DinurMR09}, but ours is
unconditional.

Turning to upper bounds for the Approximate Graph Coloring Problem, we
observe that the well-known algorithm due to Wigderson
\cite{Wigderson1983} that properly colors any~3-colorable graph
with~$n$ vertices with~$O(\sqrt{n})$ colors is implementable in
width~$4$. We generalize Wigderson's algorithm to show that, for any
fixed real~$\epsilon$ in the interval~$(0,1/2)$, the problem of
distinguishing~$3$-colorable from non-$O(n^\epsilon)$-colorable graphs
with~$n$ vertices can be solved in width~$n^{1-2\epsilon}$. This leads
to a simple algorithm that properly colors any~$3$-colorable graph
with~$O(n^\epsilon)$ colors in
time~$2^{\Theta(n^{1-2\epsilon}\log(n))}$. Asymptotically, this beats
the straightforward~$2^{\Theta(n^{1-\epsilon})}$ time-bound that
follows from partitioning the graph into~$O(n^\epsilon)$ many
independent parts each of size~$O(n^{1-\epsilon})$. As a nearly
matching lower bound, we show that the same problem cannot be solved
in width less than~$n^{1-3\epsilon}$. The problem of closing the gap
between the~$1-3\epsilon$ in the lower bound and the~$1-2\epsilon$ in
the upper bound is left as an intriguing open problem.

\section{Preliminaries} \label{sec:prelims}

For an integer~$n$, we shall use~$[n]$ to denote the
set~$\{1,2,\dots,n\}$. 

\paragraph{Tuples and relations}
Let~$A$ be a set and let~$k$ be a positive integer. A~\emph{$k$-tuple
  over~$A$} is as a sequence~$\bt = (\bt(1),\ldots,\bt(k))$,
where~$\bt(1),\ldots,\bt(k)$ are elements of~$A$; equivalently,
a~$k$-tuple over~$A$ can be seen as a map~$\bt : [k] \to A$ with
domain~$[k]$ and range in~$A$. The set of~$k$-tuples over~$A$ is
denoted by~$A^k$.  If~$\bj = (\bj(1),\ldots,\bj(m))$ is an~$m$-tuple
over~$[k]$ and~$\ba$ is a~$k$-tuple, then the
\emph{projection}~$\pr_{\bj} \ba$ of~$\ba$ on~$\bj$ is
the~$m$-tuple~$(\ba(\bj(1)),\ldots,\ba(\bj(m)))$.
If~$J = \{ j_1,\ldots,j_m \}$ is a subset of~$[k]$
with~$j_1 < \cdots < j_m$, then the projection~$\pr_J \ba$ of~$\ba$
on~$J$ is the~$m$-tuple~$(\ba(j_1),\ldots,\ba(j_m))$.

A \emph{relation}~$R$ of arity~$k$ over~$A$ is a
subset~$R \subseteq A^k$ of~$k$-tuples over~$A$.  For every
non-negative integer~$r$, every relation~$R \subseteq A^r$ of
arity~$r$ over~$A$, and every~$k$-tuple or~$k$-subset~$J$ over~$[r]$,
the \emph{projection}~$\pr_J R$ of~$R$ on~$J$ is the~$k$-ary
relation~$\{\pr_J \ba \mid \ba\in R\}$. For any two binary
relations~$R \subseteq A^2$ and~$S\subseteq A^2$ over~$A$, their
\emph{composition}~$R\circ S$ is the binary
relation~$\{(a,b) \mid \text{ there exists } c \in A \text{ such that
} (a,c)\in R \text{ and } (c,b)\in S \}$. For any two
relations~$R \subseteq A^r$ and~$S \subseteq A^s$ over~$A$ of
arities~$r$ and~$s$, their \emph{product}~$R \times S$ is the
relation~$\{ \bt \in A^{r+s} \mid \pr_{[r]} \bt \in R \text{ and }
\pr_{[r+s]\setminus[r]} \bt \in S \}$. The iterated product of more
than two relations~$R_1,\ldots,R_m$ is denoted by~$\prod_{i=1}^m R_i$.
If all~$R_i$ are the same relation~$R$, then we call it the~$m$-power
of~$R$ and denote it by~$R^m$.

Let~$f:A\rightarrow B$ be a map. We shall use~$\dom(f)$ to denote its
domain~$A$ and~$\img(f)$ to denote its image~$f(A)$. It will be
convenient to allow functions with empty domain (but note that there
is a unique function with empty domain).  If~$X\subseteq \dom(f)$ then
we shall use~$f|_{X}$ to denote the {\em restriction} of~$f$ to~$X$,
i.e., the unique map~$g$ with~$\dom(g)=X$ that agrees with~$f$
on~$X$. If~$g=f|_{X}$ for some~$X$ we shall say that~$f$ is an
\emph{extension} of~$g$. We write~$g \subseteq f$ to denote the fact
that~$f$ is an extension of~$g$. For any~$k$-tuple~$\ba\in A^k$ we
shall use~$f(\ba)$ to denote the~$k$-ary tuple obtained by
applying~$f$ to~$\ba$ component-wise,
i.e.,~$f(\ba) = (f(\ba(1)),\ldots,f(\ba(k)))$.  For two finite
sets~$A$ and~$B$ and a non-negative integer~$k$, we write~$M_k(A,B)$
to denote the set of maps from a subset of~$A$ into a subset of~$B$
with a domain of cardinality at most~$k$. For each~$f \in M_k(A,B)$,
let~$v_f = (v_{f,a,b} \mid (a,b) \in A \times B)$ denote the
\emph{characteristic vector} of~$f$, i.e.,~$v_{f,a,b} = 1$
if~$a \in \dom(f)$ and~$f(a) = b$, and~$v_{f,a,b} = 0$ otherwise, for
all~$a \in A$ and~$b \in B$.

\paragraph{Relational structures}
A \emph{signature} is a finite collection of relation symbols~$R$,
each of them with an associated non-negative integer called the~{\em
  arity} of~$R$.
If~$\sigma$ is a signature, then a~\emph{$\sig$-structure}~$\bA$
consists of a set~$A$, called the \emph{domain} of~$\bA$, and a
relation~$R^{\bA}\subseteq A^r$ for each~$R$ in~$\sig$, where~$r$ is
the arity of~$R$, called the \emph{interpretation} of~$R$ in~$\bA$. We
shall use the same capital letter to denote the universe of
a~$\sig$-structure, e.g.,~$A$ is the universe of~$\bA$. All signatures
and structures in this paper are assumed to be finite, i.e., they have
a finite number of relations over a finite domain.

Let~$\sig$ be a signature. For any two~$\sig$-structures~$\bA$
and~$\bB$ with domains~$A$ and~$B$, their {\em union}~$\bA \cup \bB$
is the~$\sig$-structure with domain~$A\cup B$ and
interpretations~$R^{\bA\cup\bB}=R^{\bA}\cup R^{\bB}$ for
all~$R \in \sigma$.  We say that~$\bA$ is a \emph{substructure}
of~$\bB$ if~$A \subseteq B$ and~$R^{\bA} \subseteq R^{\bB}$ hold for
all~$R \in \sigma$. If, furthermore,~$R^{\bA} = R^{\bB} \cap A^r$ for
all~$R \in \sigma$, where~$r$ is the arity of~$R$, then we say
that~$\bA$ is an \emph{induced substructure} of~$\bB$. We say that the
set~$A$ \emph{induces}~$\bA$ in~$\bB$ and write~$\bB{|_A}$ to denote
the substructure of~$\bB$ \emph{induced by}~$A$.  

For every
integer~$n \geq 1$ and every~$\sig$-structure~$\bA$,
the~\emph{$n$-power} of~$\bA$ is the~$\sigma$-structure~$\bA^n$ with
domain~$A^n$ and
interpretations~$R^{\bA^n} = \{ (\bt_1,\ldots,\bt_r) \in (A^n)^r \mid
(\bt_1(i),\ldots,\bt_r(i)) \in R^{\bA} \text{ for all } i \in [n] \}$
for all~$R \in \sig$, where~$r$ is the arity of~$R$.  Note
that~$R^{\bA^n}$ is \emph{not} literally the same relation as
the~$n$-power~$(R^{\bA})^n$ of the relation~$R^{\bA}$, but the two
relations are the same up to flattening the tuples as elements
of~$A^{rn}$ and permuting the coordinates.

\paragraph{Homomorphisms, CSPs and PCSPs}
Let~$\sig$ be a signature and let~$\bA$ and~$\bB$ be~$\sig$-structures
with domains~$A$ and~$B$. A map~$f:A\rightarrow B$ is a
\emph{homomorphism} from~$\bA$ to~$\bB$ if, for all~$R\in\sigma$
and~$\ba\in A^r$, where~$r$ is the arity of~$R$, it
holds~$ \ba\in R^{\bA}$ implies~$f(\ba)\in R^{\bB}$.  We shall
use~$\bA\rightarrow\bB$ to denote the statement that there exists an
homomorphism from~$\bA$ to~$\bB$.

We write~$\CSP(\bA)$ to denote the following computational problem:
\emph{given an input~$\sig$-structure~$\bI$, decide
  whether~$\bI \rightarrow \bA$ or~$\bI \not\rightarrow \bA$}.  The
structure~$\bA$ is the \emph{CSP template} of the
problem~$\CSP(\bA)$. If~$\bA \rightarrow \bB$, we
write~$\PCSP(\bA,\bB)$ to denote the following computational promise
problem: \emph{given an input~$\sig$-structure~$\bI$, distinguish the
  case~$\bI \rightarrow \bA$ from the case~$\bI \not\rightarrow \bB$,
  provided one of these cases holds (if not, any answer is valid).}
The pair of structures~$(\bA,\bB)$ is the \emph{PCSP template} of the
problem~$\PCSP(\bA,\bB)$. Since it is assumed
that~$\bA \rightarrow \bB$, note that either~$\bA$ is the empty
structure, or~$\bB$ cannot be trivial in the sense that not all the
relations of~$\bB$ can be empty.  Note also that the definition of the
problem~$\PCSP(\bA,\bB)$ assumes that the input structure~$\bI$
satisfies~$\bI\rightarrow\bA$ or~$\bI\not\rightarrow\bB$, and that
these cases are disjoint since~$\bA \rightarrow \bB$. Finally, observe
that~$\PCSP(\bA,\bA)$ is precisely the same problem as~$\CSP(\bA)$. To
simplify matters we shall take the liberty to use~$\bA$ to denote
the~$\PCSP$ template~$(\bA,\bA)$.

\paragraph{Minors and minions}

Most of the terminology that follows comes from \cite{BBKO}. Let~$m$
and~$n$ be positive integers and let~$A$ and~$B$ be finite sets.
An~$n$-ary function~$f:A^n\rightarrow B$ is called the {\em minor} of
an~$m$-ary function~$g:A^m\rightarrow B$ given by the
map~$\pi:[m]\rightarrow[n]$ if the following identity holds
\begin{equation}
f(x_1,\dots,x_n)\approx g(x_{\pi(1)},\dots,x_{\pi(m)}),
\end{equation}
i.e., if the
equality~$f(x_1,\dots,x_n)=g(x_{\pi(1)},\dots,x_{\pi(m)})$ holds for
all~$x_1,\dots,x_n\in A$. Informally, one can say that~$f$ is a minor
of~$g$ if it can be obtained from~$g$ by permuting variables,
identifying variables, and introducing dummy variables. We shall
use~$g_{\pi}$ to denote the minor of~$g$ given by the map~$\pi$.  A
\emph{minion} on~$(A,B)$ is any non-empty subset
of~$\{f:A^n\rightarrow B \mid n\geq 1\}$ that is closed under
minors. Let~$\script{M}$ and~$\script{N}$ be minions (not necessarily
on the same pair of sets). A {\em minion homomorphism}
from~$\script{M}$ to~$\script{N}$ is any
mapping~$\xi:\script{M}\rightarrow \script{N}$ satisfying the
following conditions:
\begin{enumerate} \itemsep=0pt
\item it preserves arities, i.e., for every $g$ in $\script{M}$, its image $\xi(g)$ has the same arity as $g$,
\item it preserves taking minors, i.e., for all integers~$m$ and~$n$, all
  maps~$\pi : [m]\rightarrow [n]$ and all~$m$-ary functions~$g$
  in~$\script{M}$, we
  have~$\xi(g)_\pi=\xi(g_\pi)$.
\end{enumerate}

Let~$\sigma$ be a signature, let~$\bA$ be a~$\sig$-structure, and
let~$\script{M}$ be a minion (not necessarily related to~$\bA$). The
{\em free structure} of~$\bA$ generated by~$\script{M}$ is
the~$\sig$-structure~$\bF(\bA; \script{M})$ defined as follows.
Let~$n = |A|$, regard the domain~$A=\{x_1,\dots,x_n\}$ of~$\bA$ as a
collection of variables, and define the universe
of~$\bF(\bA; \script{M})$ to be the set of~$n$-ary functions
in~$\script{M}$. Let~$R\in\sig$ be any relation symbol, let~$r$ be its
arity, and let~$\bt_1,\dots,\bt_m$ be an arbitrary ordering of the
tuples in~$R^{\bA}$. For each~$i \in [r]$,
let~$\pi_i : [m] \rightarrow [n]$ be the map defined
by~$\pi_i(j) = \bt_j(i)$, for~$j \in [m]$.
Then,~$R^{\bF(\bA; \script{M})}$ is defined to contain, for
every~$m$-ary function~$g\in{\script{M}}$, the
tuple~$(f_1,\dots,f_r)$, where~$f_i(x_1,\dots,x_n)$ is the~$n$-ary
minor of~$g$ given by the map~$\pi_i : [m] \rightarrow [n]$; i.e., the
function~$f_i$ is the unique~$n$-ary map that satisfies the identity
\begin{equation}
f_i(x_1,\ldots,x_n) \approx g(\bt_1(i),\dots,\bt_m(i))
\end{equation}
over the base domain of~$g$.  We note here that since every
minion~$\script{M}$ is closed under permuting the variables of a
function, the structure~$\bF(\bA; \script{M})$ is well defined in the
sense that it does not depend on the choice of ordering of the
elements in the domain~$A$ of~$\bA$, or on the choice for the ordering
of the tuples of its relations.

\begin{example}
  \label{ex:WNU}
Let~$m\geq 3$ be an integer and let~$\bA$ be the structure with
domain~$\{x,y\}$ with a single~$m$-ary
relation~$R^{\bA}=\{(y,x,\dots,x),(x,y,\dots,x),\dots,(x,x,\dots,y)\}$;
the set of all tuples with exactly one occurrence of~$y$.  By
construction, for every minion~$\script{M}$ and every~$m$-ary
function~$g \in \script{M}$, the relation~$R^{\bF(\bA; {\script{M}})}$
contains the tuple~$(f_1,\ldots,f_m)$ where~$f_i$ is the binary
operation defined by the identity
\begin{equation}
f_i(x,y) \approx g(x,\dots,x,y,x,\dots,x), \label{deffi}
\end{equation} 
where~$y$ appears in position~$i$ in the
tuple~$(x,\dots,x,y,x,\dots,x)$ on the right-hand side
of~\eqref{deffi}. Now, if $g$ satisfies the identities
\begin{equation}
g(y,x,x,\dots,x)\approx g(x,y,x,\dots,x)\approx \cdots \approx g(x,\dots,x,x,y), \label{eq:1}
\end{equation}
then~$f_1 = \cdots = f_m$ and~$R^{\bF(\bA; \script{M})}$ contains the
reflexive tuple~$(f,f,\dots,f)$ where~$f := f_1 = \cdots =
f_m$. Conversely, if~$R^{\bF(\bA; {\script{M}})}$ contains a reflexive
tuple, then there exists some~$r$-ary function~$g\in{\script{M}}$
satisfying~\eqref{eq:1}. Any function of arity~$m$
satisfying~\eqref{eq:1} is called an~$m$-ary \emph{weak near
  unanimity}~(WNU).
\end{example}

\begin{observation}\label{obs:homtofree}
Let $\bA$ be a $\sigma$-structure and $\script{M}$ be a minion. Then $\bA\rightarrow \bF(\bA; \script{M})$
\end{observation}
\begin{proof}
  Assume that the universe of~$\bA$ is~$\{x_1,\dots,x_n\}$ as
  usual. Since minions are nonempty and closed under identification of
  variables it follows that~$\script{M}$ contains a unary
  operation~$g$. Hence, it must contain also, for every~$i\in[n]$,
  the~$n$-ary operation~$f_i$ defined by the
  identity~$f_i(x_1,\dots,x_n)\approx g(x_i)$. Define the
  mapping~$\varphi:A\rightarrow F(\bA; \script{M})$
  as~$\varphi(x_i)=f_i$, for~$i\in[n]$. We shall show that~$\varphi$
  defines an homomorphism from~$\bA$ to~$\bF(\bA;
  \script{M})$. Let~$R\in\sigma$ and let~$\bt_1,\dots,\bt_m$ be the
  ordering of the tuples in~$R^{\bA}$ used in the definition
  of~$R^{\bF(\bA; \script{M})}$. It remains to be shown
  that~$\varphi(\bt_i)\in R^{\bF(\bA; \script{M})}$ for
  every~$i\in[m]$. Indeed, for every~$i\in[m]$, let~$g^m_i$ be
  the~$m$-ary function defined by the identity
  $$g^m_i(y_1,\dots,y_m)\approx g(y_i).$$ 
  Note that~$g^m_i$ is in~$\script{M}$ as it is a minor of~$g$. To
  complete the proof it suffices to observe that the tuple
  included in~$R^{\bF(\bA; \script{M})}$ due to~$g^m_i$ is
  precisely~$\varphi(\bt_i)$.
\end{proof}

\paragraph{Polymorphisms}

Let~$\sig$ be a signature, let~$(\bA,\bB)$ be a~$\PCSP$ template of
signature~$\sig$, and let~$n$ be a positive integer. An~$n$-ary {\em
  polymorphism} of~$(\bA,\bB)$ is a homomorphism from~$\bA^n$
to~$\bB$; that is, unfolding the definitions, a
mapping~$f:A^n\rightarrow B$ such that, for all~$R\in\sig$
and~$\bt_1,\dots,\bt_n\in R^{\bA}$, it holds
that~$f(\bt_1,\dots,\bt_n)\in R^{\bB}$, where
\begin{equation}
f(\bt_1,\dots,\bt_n) =
(f(\bt_1(1),\dots,\bt_n(1)),\ldots,f(\bt_1(r),\dots,\bt_n(r)))
\end{equation}
and~$r$ is the arity of~$R$. We denote the set of all polymorphisms
of~$(\bA,\bB)$ by~$\Pol(\bA,\bB)$. As usual we shall use~$\Pol(\bA)$
as a shorthand for~$\Pol(\bA,\bA)$.  It follows from the definitions
that the collection~$\Pol(\bA,\bB)$ of all polymorphisms
of~$(\bA,\bB)$ is a minion.

The following result, pointed out in \cite{BBKO}, will be useful.
\begin{observation}\cite{BBKO}
   \label{obs:freeminorhom}
 For every minion~${\script{M}}$ and every structure~$\bA$ there exists
 a minor homomorphism $\xi$ from~${\script{M}}$
 to~$\Pol(\bA,\bF(\bA; {\script{M}}))$.
 \end{observation}
 \begin{proof}
   To simplify notation assume again
   that~$A=\{x_1,\dots,x_n\}$. Then~$\xi$ maps every
   operation~$g\in\script{M}$ with arity, say,~$m$, to the
   operation~$\xi(g)\in \Pol(\bA,\bF(\bA; {\script{M}}))$ defined
   by~$\xi(g)(x_{\pi(1)},\dots,x_{\pi(m)})=g_{\pi}$ for
   every~$\pi : [m] \to [n]$, where~$g_{\pi}$ is the~$n$-ary minor
   of~$g$ given by map~$\pi$, i.e.,~$g_{\pi}$ satisfies the identity
$$g_{\pi}(x_1,\dots,x_n)\approx g(x_{\pi(1)},\dots,x_{\pi(m)})$$
over the base domain of~$g$.  It can be readily verified that~$\xi$
thus defined defines a minor homomorphism from~${\script{M}}$
to~$\Pol(\bA,\bF(\bA; {\script{M}}))$.
\end{proof}

\section{Relaxations for (P)CSP}
 
A common heuristic method to solve an instance~$\bI$
of~$\PCSP(\bS,\bT)$ consists in solving a polynomial-time solvable
relaxation of~$\bI\rightarrow\bS$. If the relaxation turns out to be
not feasible, then we can infer for sure
that~$\bI\not\rightarrow\bS$. Although the converse is not necessarily
true, for some templates~$(\bS,\bT)$ it can be guaranteed that, if the
relaxation is feasible, then~$\bI\rightarrow\bT$. For such templates,
the heurisitic method is a valid polynomial-time algorithm
for~$\PCSP(\bS,\bT)$. In the present paper we will focus mostly on
relaxations based on local consistency.

\subsection{Local consistency algorithm} \label{sec:consistency}

Fix a signature~$\sigma$. Let~$\bI$ be an instance with domain~$I$,
let~$\bS$ be a constraint language with domain~$S$, and let~$k$ be a
positive integer. A mapping $f:X\rightarrow S$ with $X\subseteq I$ is
a {\em partial} homomorphism from $\bI$ to $\bS$ if it is an homomorphism from $\bI{|_X}$
to $\bS$. A {\em~$k$-strategy on~$\bI$ and~$\bS$}
\cite{KolaitisV08} is any nonempty collection~$\script{H}$ of partial
homomorphisms from~$\bI$ to~$\bS$ such that:
\begin{enumerate} \itemsep=0pt
\item the family $\script{H}$ it is closed under restrictions, i.e.,
  for every~$h$ in~$\script{H}$ and every~$X\subseteq\dom(h)$, the
  restriction~$h|_X$ of~$h$ to~$X$ is in~$\script{H}$,
\item the family~$\script{H}$ has the extension property up to~$k$,
  i.e., for every~$h$ in~$\script{H}$ with~$|\dom(h)|<k$ and
  every~$x\in I\setminus\dom(h)$, there exists~$f$ in~$\script{H}$
  such that~$h \subseteq f$ and~$\dom(f)=\dom(h)\cup\{x\}$.
\end{enumerate}
We write~$\bI \leq_k \bS$ to denote the statement that there exists
a~$k$-strategy on~$\bI$ and~$\bS$. It follows directly from
the definitions that, if~$\bI \to \bS$, then~$\bI \leq_k \bS$. Indeed,
if~$h$ is a homomorphism from~$\bI$ to~$\bS$, then the
collection~$\script{H} = \{h|_X \mid X\subseteq I, |X|\leq k\}$ is
a~$k$-strategy on~$\bI$ and~$\bS$.

There is an algorithm that, given~$(\bI,\bS,k)$ as input, decides
whether there exists a~$k$-strategy and does so in time polynomial
in~$(|I|+|S|)^k$. This algorithm, usually called
the~$(k-1)$-consistency algorithm, starts by placing in~$\script{H}$
all partial homomorphisms from~$\bI$ to~$\bS$ with domain size at
most~$k$ and repeatedly removes from~$\script{H}$ those~$h$ that
falsify any one of the two conditions in the definition
of~$k$-strategy. The algorithm stops when it reaches a fixed-point. If
the fixed-point obtained is non-empty, then it necessarily is
a~$k$-strategy. Otherwise, it can be safely concluded that
no~$k$-strategy exists.  Since, as seen above,~$\bI \to \bS$
implies~$\bI \leq_k \bS$, one could use the consistency algorithm as a
partial-check to decide whether~$\bI$ is homomorphic to~$\bS$.  This
is the basis of the \emph{width heuristic} for CSP which, following
\cite{BBKO}, we now extend to PCSPs.

Let~$(\bS,\bT)$ be a PCSP template of signature~$\sig$; i.e.,~$\bS$
and~$\bT$ are~$\sig$-structures such that~$\bS \to \bT$.
Let~$k = k(n)$ be an integer function. We say that~$\PCSP(\bS,\bT)$
\emph{is solvable in width~$k(n)$}, or that the template~$(\bS,\bT)$
\emph{has width~$k(n)$}, if for every~$\sig$-structure instance~$\bI$
with~$n$ elements it holds that~$\bI \leq_{k(n)} \bS$
implies~$\bI \to \bT$.  Note that this generalizes the definition of
width for CSPs since, whenever~$\PCSP(\bS,\bS) = \CSP(\bS)$ is
solvable in width~$k(n)$, the condition~$\bI \leq_{k(n)} \bS$ is
necessary \emph{and} sufficient for~$\bI \to \bS$. We note that
every~$\PCSP(\bS,\bT)$ is solvable in width at most~$k(n)$
for~$k(n) = n$. Whenever it is solvable in width~$k(n)$ for
some~$k(n) = O(1)$, we say that~$\PCSP(\bS,\bT)$ is solvable \emph{in
  bounded width}, or that~$(\bS,\bT)$ \emph{has} bounded
width. Whenever it is solvable in width~$k(n)$ for some~$k(n) = o(n)$,
we say that~$\PCSP(\bS,\bT)$ is solvable in \emph{sublinear width}, or
that~$(\bS,\bT)$ \emph{has} sublinear width.

In the particular case of CSPs, the strength of bounded width is well
understood. The breakthrough of Barto and Kozik \cite{BartoK14}, in
combination with \cite{Barto2016Collapse}, yields the following
characterization of CSP templates of bounded width:

\begin{theorem}\label{th:width} \textup{\cite{BartoK14,Barto2016Collapse}}
For every structure $\bT$ the following are equivalent:
\begin{enumerate} \itemsep=0pt
\item $\bT$ has bounded width,
\item $\bT$ has width $\max(r,3)$ where $r$ is the maximum arity in $\bT$,
\item $\Pol(\bT)$ contains an~$m$-ary WNU for every integer~$m\geq 3$.
\end{enumerate}
\end{theorem}

\noindent We illustrate Theorem~\ref{th:width} with an example that will
become useful in Section~\ref{sec:approximatecoloring}.

\begin{example}\label{ex:majority}
  Consider the problem of deciding whether a graph~$\bG$
  is~$2$-colorable, i.e,~$\CSP(\clique_2)$. If~$\bG$ contains a
  cycle~$a_0,a_1,\dots,a_n=a_0$ of odd length~$n$,
  then there is no~$k$-strategy on~$\bG$ and~$\bK_2$: an easy
  inductive argument shows that every such strategy would need to
  contain for every~$j\geq 1$ a partial homomorphism~$f$
  with domain~$\{a_0,a_j\}$ satisfying~$f(a_0)=f(a_j)$
    if~$j$ is even and~$f(a_0)\neq f(a_j)$ if~$j$ is odd, which is not
    possible because~$n$ is odd and~$a_n = a_0$.
  Consequently,~$\clique_2$ has width~$3$.

  One can alternatively look at this algebraically and note that, for
  every integer~$m\geq 3$, the set~$\Pol(\clique_2)$ of polymorphisms
  of~$\clique_2$ contains the function~$\varphi:[2]^m\rightarrow[2]$
  that returns the majority of its arguments. These are WNUs and,
  therefore, by virtue of Theorem~\ref{th:width}, it follows also
  that~$\clique_2$ has width~$3$. Indeed, it was already observed
  in~\cite{FederVardi1998} that a structure~$\bT$ has bounded width
  whenever~$\Pol(\bT)$ contains a majority operation~$f$, i.e, an
  operation~$f:S^3\rightarrow S$ satisfying
  \begin{equation}
  f(y,x,x)\approx f(x,y,x)\approx f(x,x,y) \approx x \label{eqn:majid}
\end{equation}
Note that the first two identities in~\eqref{eqn:majid} alone are the
WNU identities in~\eqref{eq:1} for arity~$3$.
\end{example}

Let us now include sublinear width into the picture. It has been known
that Theorem~\ref{th:width} can be strengthened to also include
\emph{sublinear width} in the characterization. Concretely,
items~\emph{1}--\emph{3} are also equivalent to item~\emph{0} below:
\begin{enumerate}
\item[\emph{0}.] \emph{$\bT$ has sublinear width.}
\end{enumerate}
On one hand, clearly~\emph{1} implies~\emph{0}. On the other hand, it
is known that, if~\emph{3} fails, then~$\bT$ is able to
\emph{simulate} the constraint language corresponding to systems of
equations over a non-trivial finite Abelian group, for which linear
width lower bounds are known (see,
e.g.,~\cite{AtseriasO19}). Thus, \emph{0} implies
\emph{3}. As it turns out, this implication will also follow from our
main result in Section~\ref{sec:mainresult}.

As in the CSP world, whether a fixed-template PCSP is solvable in
bounded width is controlled by the polymorphisms of the
template~\cite{BBKO}.

\begin{theorem}
\label{th:minorwidth} \textup{\cite{BBKO}}
Let~$(\bS_1,\bT_1)$ and~$(\bS_2,\bT_2)$ be two PCSP templates such
that there exists a minor homomorphism from~$\Pol(\bS_1,\bT_1)$
to~$\Pol(\bS_2,\bT_2)$. Then, the following statements hold:
\begin{enumerate} \itemsep=0pt
\item If $(\bS_1,\bT_1)$ has bounded width, then so does $(\bS_2,\bT_2)$.
\item If $(\bS_1,\bT_1)$ has sublinear width, then so does $(\bS_2,\bT_2)$.
\end{enumerate}
\end{theorem}

\noindent
We note here that, although not explicitly addressed in~\cite{BBKO},
the proof for the bounded width statement in
Theorem~\ref{th:minorwidth} works also for sublinear width; i.e., the
proof of Item~\emph{1} in Theorem~\ref{th:minorwidth} also proves
Item~\emph{2}.

\subsection{Linear programming and Sherali-Adams hierarchy}

In this subsection we shall introduce a more general family of
relaxations based on linear programming. Despite the fact that,
in general, these relaxations are seemingly more
powerful, in the realm of fixed-template CSPs they
turn out to be not more powerful than local consistency. However, it
will follow from our results that for fixed-template
  PCSPs they are, indeed, strictly more powerful.

Fix a signature~$\sigma$. Let~$\bI$ be an instance with
domain~$I = [n]$, and let~$\bS$ be a constraint language with
domain~$S = [p]$, where~$n$ and~$p$ are positive integers. For
each~$R \in \sig$ of arity~$r$,
let~$Z_R \subseteq \{0,1\}^{rp}$ denote the set of characteristic
vectors~$v_{\bt}$ as~$\bt$ ranges over the set of~$r$-tuples
in~$R^{\bS}$, seen as maps from~$[r]$ to~$[p]$; formally,
\begin{equation}
v_{\bt} =
(v_{\bt,1,1},\ldots,v_{\bt,1,p},\ldots,v_{\bt,r,1},\ldots,v_{\bt,r,p}),
\end{equation}
where~$v_{\bt,i,j} = 1$ if~$\bt(i) = j$ and~$v_{\bt,i,j} = 0$
if~$\bt(i) \not= j$, for all~$i \in [r]$ and~$j \in [p]$.
Let~$P_R \subseteq \mathbb{R}^{rp}$ denote the convex hull
of~$Z_R \subseteq \{0,1\}^{rp}$. The polytope~$P_R$ can be described
by a linear program of the form~$A_R x \leq b_R$, where~$x$ is a
vector of~$rp$ many variables,~$A_R$ is an integer matrix with at
most~$2^{rp}$ many rows, and~$b_R$ is an integer vector of dimension
at most~$2^{rp}$.  Let~$\mathrm{IP}(\bI,\bS)$ denote the 0-1 linear
program that has one variable~$x_{u \mapsto a}$ for each~$u \in I$
and~$a \in S$ and the following defining equalities and inequalities:
\begin{center}
\begin{tabular}{llll}
$x_{u \mapsto 1} + \cdots + x_{u \mapsto p} = 1$ & & for $u \in I$, & (I1) \\
$A_R x_{\bu} \leq b_R$ & & for $R \in \sigma$ and $\bu \in R^{\bI}$, & (I2) \\
$x_{u \mapsto a} \in \{0,1\}$ & & for $u \in I$ and $a \in S$, & (I3)
\end{tabular}
\end{center}
where, for $\bu = (u_1,\ldots,u_r)$, we define
\begin{equation}
x_{\bu} := (x_{u_1 \mapsto 1},\ldots,x_{u_1 \mapsto p},\ldots,
x_{u_r \mapsto 1},\ldots,x_{u_r \mapsto p}).
\end{equation}
The \emph{direct LP relaxation} of~$\mathrm{IP}(\bI,\bS)$, denoted
by~$\mathrm{LP}(\bI,\bS)$, is the linear program that is obtained by
replacing the 0-1 constraints~(I3) by their
relaxation~$0 \leq x_{u \mapsto a} \leq 1$.  We note
that~$\mathrm{LP}(\bI,\bS)$ is equivalent to the LP obtained by taking
the so-called {\em basic linear programming} relaxation introduced in
the optimization variants MAX-CSP of CSP (see,
e.g.,~\cite{BrakensiekGWZ20,Butti2021,
  DalmauK13,DalmauKM18,GhoshT18,KumarMTV11,Kun2012}), and turning the
objective function into a family of inequalities.  This is a weak
relaxation. In order to obtain stronger LP relaxations we apply the
method of Sherali and Adams~\cite{SheraliA90}
to~$\mathrm{IP}(\bI,\bS)$.

Fix an integer~$k \geq 1$. Recall that~$M_k(I,S)$ is used to denote
the set of maps from a subset of~$I$ of cardinality at most~$k$
into~$S$.  Let~$\SA^k(\bI,\bS)$ denote the linear program that has one
variable~$x_f$ for each~$f \in M_k(I,S)$ and the following defining
equalities and inequalities:
\begin{center}
\begin{tabular}{llll}
  $x_\emptyset = 1$ & & & (T1) \\
  $x_{f \cup \{u \mapsto 1\}} + \cdots + x_{f \cup \{u \mapsto p\}} = x_f$ & & for $f \in M_{k-1}(I,S)$ and $u \in I \setminus \dom(f)$, & (T2) \\
  $A_R x_{f,\bu} \leq b_R x_f$ & & for $f \in M_{k-1}(I,S)$, $R \in \sig$ and $\bu \in R^{\bI}$, & (T3) \\
  $0 \leq x_f \leq 1$ & & for $f \in M_k(I,S)$, & (T4)
\end{tabular}
\end{center}
where, for~$f \in M_{k-1}(I,S)$ and~$\bu = (u_1,\ldots,u_r)$, we define
\begin{align}
  & x_{f,\bu} := 
    (x_{f \cup \{u_1 \mapsto 1\}},\ldots,x_{f \cup \{u_1 \mapsto p\}},
    \ldots,x_{f \cup \{u_r \mapsto 1\}},\ldots,x_{f \cup \{u_r \mapsto p\}}).
\end{align}
It is obvious that~$\SA^1(\bI,\bS)$ is basically the same linear
program as~$\mathrm{LP}(\bI,\bS)$: in case~$k = 1$, the
set~$M_{k-1}(I,S)$ contains only the empty map~$f = \emptyset$,
and~$x_\emptyset = 1$ by~(T1) in~$\SA^1$. Furthermore,
since~$M_k(I,S) \subseteq M_{k+1}(I,S)$, it is clear that
each~$\SA^{k+1}(\bI,\bS)$ is at least as strong
as~$\SA^k(\bI,\bS)$. To understand the relaxation, think of each new
inequality of~$\SA^{k+1}(\bI,\bS)$ as obtained from multiplying an
inequality in~$\SA^k(\bI,\bS)$ by a variable~$x_{u \mapsto a}$,
simplifying the results by the
rules~$x_{u \mapsto a}^2 = x_{u \mapsto a}$
and~$x_{u \mapsto a} x_{u \mapsto b} = 0$ whenever~$a \not= b$ (these
relations are true in all solutions of~$\mathrm{IP}(\bI,\bS)$), and
finally linearizing all the quadratic terms by introducing new
variables for the products.

It is worth noting that, according to the Sherali-Adams method as
defined in~\cite{SheraliA90}, it would seem that our definition
of~$\SA^{k+1}(\bI,\bS)$ is missing the inequalities that can be
obtained from those in~$\SA^k(\bI,\bS)$ from multiplication
by~$1-x_{u \mapsto a}$. However, we note that these multiplications
are redundant: the equalities of type~(T1) and~(T2) imply
that~$1-x_{u \mapsto a} = \sum_{b \in [p] \setminus \{a\}} x_{u
  \mapsto b}$, which means that a multiplication
by~$1-x_{u \mapsto a}$ can be obtained as a positive linear
combination of multiplications by~$x_{u \mapsto b}$ for~$b \not= a$.

We write~$\bI \leq_{\SA^k} \bS$ if the linear program~$\SA^k(\bI,\bS)$
is feasible. It follows directly from the definitions that
if~$\bI \to \bS$, then~$\bI \leq_{\SA^k} \bS$.  Indeed, if~$h$ is a
homomorphism from~$\bI$ to~$\bS$, then the assignment that
sets~$x_f = 1$ if~$f \subseteq h$, and~$x_f = 0$
if~$f \not\subseteq h$, is a feasible solution to the linear
program. Furthermore, it also follows directly from the definitions
that if~$\bI \leq_{\SA^k} \bS$, then~$\bI \leq_k \bS$; indeed,
if~$(x_f \mid f \in M_k(I,S))$ is a feasible solution
for~$\SA^k(\bI,\bS)$, then the
collection~$\script{H} = \{ f \in M_k(I,S) \mid x_f \not= 0 \}$ is
a~$k$-strategy for~$\bI$ and~$\bS$; the non-emptiness condition
follows from~(T1), the closure under restrictions and the extension
property up to~$k$ follow from~(T2), and the condition that every map
in~$\script{H}$ is a partial homomorphism follows from~(T3) and the
choice of the polytope~$P_R$.

Since the feasibility problem for linear programs is solvable in
polynomial time, there is an algorithm that, given~$(\bI,\bS,k)$ as
input, decides whether~$\bI \leq_{\SA^k} \bS$ and does so in time
polynomial in~$(|I|+|S|)^k$. This algorithm is the \emph{width-$k$
  Sherali-Adams heuristic} for~CSPs which, as was the case for the
width-$k$ consistency algorithm, can also be used as a heuristic for
PCSPs. If~$(\bS,\bT)$ is a PCSP template and~$k(n)$ is an integer
function, then we say that~$\PCSP(\bS,\bT)$ \emph{is solvable in
  SA-width~$k(n)$}, or that the template~$(\bS,\bT)$ has
\emph{SA-width~$k(n)$}, if for every~$\sig$-structure instance~$\bI$
with~$n$ elements it holds that~$\bI \leq_{\SA^{k(n)}} \bS$
implies~$\bI \to \bT$.  It follows from the fact that it is always the
case that~$\bI \leq_n \bS$, where~$n$ is the number of elements
of~$\bI$, that every~$\PCSP(\bS,\bT)$ is solvable in~SA-width at
most~$k(n)$ for~$k(n) = n$.  Whenever it is solvable in
SA-width~$k(n)$ for some~$k(n) = O(1)$, we say that~$\PCSP(\bS,\bT)$
is solvable in \emph{bounded SA-width}, or that~$(\bS,\bT)$ has
bounded~SA-width.

The SA relaxations have been intensively used in optimization versions
of CSP~\cite{ChanLRS13,GeorgiouMT09,ThapperZ17,YoshidaZ14}.  Our setup
is, again, slightly different because the constraints of the instance
must be encoded in the polytope-defining inequalities, rather than in
the objective function, as is done in the
optimization variants. Still, the family of SA relaxations commonly
used for optimization can be readily adapted to our setting by turning
the objective function into a set of inequalities (see
\cite{Butti2021b}). The LP family obtained in this way, although not
equivalent to ours, has the same power up to constants. Indeed, the
notion of bounded SA-width remains unaltered.

The strength of SA-width~$k$ for~$k=1$ is well understood. Indeed, it
is shown in \cite{BBKO} (Theorem 7.9) that~$(\bS,\bT)$ has
SA-width~$1$ if and only if~$\Pol(\bS,\bT)$ has symmetric functions of
all arities, where an operation is symmetric if its output is
independent of the order of its input elements. The picture is also
clear if~$\bS=\bT$ (i.e., in the case of CSPs). In particular, it
follows from combining the results in~\cite{AtseriasBD09}
and~\cite{AtseriasM13} with those in~\cite{BartoK14} that
if~$\CSP(\bT)$ is solvable in bounded SA-width, then it is solvable in
bounded width as well. Furthermore, combining the results
in~\cite{BartoK14} and~\cite{Barto2016Collapse} it can be concluded
that, additionally,~$\CSP(\bT)$ is solvable in {\em relational
  width~2}, which, in turn, can be shown to imply that it is solvable
in SA-width~$2$.

\section{Main result and applications}\label{sec:mainresult}

In this section we state the main structural result about PCSP
templates of bounded width, derive from it the algebraic consequence
about the presence of WNUs of arities three and more, and apply it to
compare the relative power of bounded width and bounded SA-width.

\subsection{Structure of PCSP templates of bounded width}

If~$R \subseteq A^r$ is a relation of arity~$r$ over the set~$A$, then
we write~$\graph(R)$ to denote the set of all binary projections
of~$R$; i.e.,~$\graph(R) = \{\pr_{i,j} R \mid i,j\in[r],\, i\neq j\}$.
If~$\sig$ is a signature and~$\bA$ is a~$\sig$-structure, then we
define~$\graph(\bA) = \graph(\prod_{\sigma \in R} R^{\bA})$.  The main
technical result of this paper is the following.

\begin{theorem}
  \label{the:main}
  Let~$(\bS,\bT)$ be a PCSP template that has sublinear width. If for
  all~$U,V\in\graph(\bS)$ we
  have~$U\circ V=S^2$, then~$\bT$ is reflexive; i.e., there
  exists~$a \in T$ such that each relation in~$\bT$ contains the
  reflexive tuple~$(a,a,\dots,a)$. In particular, this holds
  if~$(\bS,\bT)$ has bounded width.
\end{theorem}

\noindent An example of a PCSP template that satisfies the condition of
Theorem~\ref{the:main} is~$(\bK_2,\bH)$ for any graph~$\bH$ having at
least one edge.  Indeed, the composition of the edge relation
of~$\bK_2$ with itself is the equality relation on~$[2]$, which is not
the full binary relation~$[2] \times [2]$. This is consistent with the
fact that~$\bK_2$, and hence~$(\bK_2,\bH)$, has width three (see
Example~\ref{ex:majority}). An example of a
PCSP template that does \emph{not} satisfy the condition of
Theorem~\ref{the:main} is~$(\bK_p,\bH)$ for any integer~$p \geq 3$ and
any self-loop free graph~$\bH$ such that~$\bK_p \to \bH$. Indeed, as
it is easy to check, if~$p \geq 3$ then the composition of the edge
relation of~$\bK_p$ with itself is the full binary
relation~$[p] \times [p]$, but~$\bH$ is not reflexive (since it is
self-loop free). We revisit these examples in
Section~\ref{sec:approximatecoloring}.

The proof of Theorem~\ref{the:main} will be given in
Section~\ref{sec:proof}. We devote the rest of this section to derive
some applications.

\subsection{Algebraic implications}

We first derive some algebraic implications.  Recall the definition of
WNU polymorphism from Example~\ref{ex:WNU} in
Section~\ref{sec:prelims}.  Recall also that, by
Theorem~\ref{th:width}, if~$\CSP(\bT)$ is solvable in bounded width,
then~$\Pol(\bT)$ must contain a WNU of every arity~$m\geq 3$. It
follows from Theorem~\ref{the:main} that this also holds for Promise
CSPs. 

\begin{corollary} \label{cor:wnus}
  Let~$(\bS,\bT)$ be a PCSP template of sublinear or bounded
  width. Then,~$\Pol(\bS,\bT)$ contains a WNU of arity~$m$ for
  every~$m \geq 3$.
\end{corollary}
\begin{proof}
  Let~$m\geq 3$, let~$\script{M}=\Pol(\bS,\bT)$, and
  let~$\bS'=(\{x,y\},R^{\bS'})$ be the structure with a~$2$-element
  domain~$\{x,y\}$ where~$R^{\bS'}$ is the~$m$-ary relation that
  contains precisely all tuples in~$\{x,y\}^m$ in which~$y$ appears
  exactly once (as in Example~\ref{ex:WNU}).  As pointed out in
  Observation~\ref{obs:homtofree}, there is a homomorphism from~$\bS'$
  to~$\bF(\bS';\script{M})$, and, hence,~$(\bS',\bF(\bS';\script{M}))$
  is a legit PCSP template.  Also, by
  Observation~\ref{obs:freeminorhom} there is a minor homomorphism
  from~$\script{M}$ to~$\Pol(\bS',\bF(\bS';\script{M}))$. Therefore,
by Theorem~\ref{th:minorwidth}, the PCSP
template~$(\bS',\bF(\bS';\script{M}))$ has sublinear width. Then, it
follows from Theorem~\ref{the:main} that~$R^{\bF(\bA';\script{M})}$
contains a reflexive tuple, which implies, by Example~\ref{ex:WNU},
that~$\script{M}$ contains an~$m$-ary WNU. 
\end{proof}

Next we argue that the converse to Corollary~\ref{cor:wnus} does not
hold. For positive integers~$s$ and~$r$, let~$\EXACTLY{s}{r}$ be the
structure with domain~$\{0,1\}$ and a single relation $R^{\EXACTLY{s}{r}}$ of arity~$r$
containing all tuples with exactly~$s$ many~$1$'s. Let~$\NAE{r}$ be
the structure with domain~$\{0,1\}$ and a single relation $R^{\NAE{r}}$ of arity~$r$
containing all tuples in~$\{0,1\}^r$ except the two reflexive
tuples~$(0,0,\ldots,0)$ and~$(1,1,\ldots,1)$.  Clearly,
if~$0 < s < r$, then there is a homomorphism from~$\EXACTLY{s}{r}$
to~$\NAE{r}$, so~$(\EXACTLY{s}{r},\NAE{r})$ is a PCSP template.

\begin{lemma}
  There is a~PCSP template~$(\bS,\bT)$ that does not have sublinear or
  bounded width such that~$\Pol(\bS,\bT)$ contains a WNU of arity~$m$
  for every~$m \geq 3$. Concretely, setting~$\bS = \EXACTLY{2}{4}$
  and~$\bT = \NAE{4}$ gives such an example.
\end{lemma}

\begin{proof}
  Let~$\bS$ and~$\bT$ be set as in the second part of the lemma.  It
  follows directly from Theorem~\ref{the:main} that~$(\bS,\bT)$ does
  not have sublinear width. However,~$\Pol(\bS,\bT)$ contains, for
  every~$m\geq 3$, the WNU~$h$ or arity~$m$, that returns the majority
  element in the input if exists and else (that is, in case of ties)
  it returns the first element. That is,~$h$ is the following
  function:
\begin{equation*}
h(x_1,\dots,x_m)=\left\{\begin{array}{ll}
a & \text{if } |\{x_i \mid x_i=a\}|>m/2 \\
x_1 & \text{otherwise} 
\end{array}\right.
\end{equation*}
To see that this is a WNU, use the assumption that~$m \geq 3$. To see
that this a polymorphism of~$(\bS,\bT)$, we argue by double
counting. Any~$m \times m$ matrix that has all its rows in~$\bT$ has
the same number of~$0$'s and~$1$'s. Therefore, either all columns also
have the same number of~$0$'s and~$1$'s, in which case~$h$ returns the
first row of the matrix, which is not reflexive, or some column has
more~$0$'s than~$1$'s and some column has more~$1$'s than~$0$'s, in
which case~$h$ returns also a non-reflexive tuple.
\end{proof}

\subsection{Separation of width from SA-width}

Another application of Theorem~\ref{the:main} is the separation of
bounded width and bounded SA-width for PCSPs. This separation is
another evidence that the family of PCSPs is a more rich family of
problems than classical CSPs from an algorithmic point of view.
This separation is obtained combining our result with the following
lemma, which in turn, builds upon \cite{BrakensiekG18}.

\begin{lemma}\label{le:ExNAE}
  Let~$\bS=\EXACTLY{s}{r}$ and~$\bT=\NAE{r}$ where~$s$ and~$r$ are
  integers such that~$0 < s < r$. Then~$(\bS,\bT)$ has SA-width~$2$.
\end{lemma}
\begin{proof}
  We start by noticing that a simple double counting argument shows
  that~$\Pol(\bS,\bT)$ contains, for every odd~$m\geq 3$, the {\em
    alternating threshold} operation~$f(x_1,\dots,x_m)$ defined to
  be~$1$ if~$\sum_{i\in[m]} (-1)^{i-1} x_i>0$ and~$0$
  otherwise. Indeed, let~$\bt_1,\dots,\bt_m$ be any sequence of tuples
  in~$R^{\bS}$ and let~$O$ and~$E$ be, respectively, the odd and even
  numbers in~$[m]$. It we use~$|\bt|_1$ to denote the total number
  of~$1$'s occurring in tuple~$\bt$, we observe that
  \begin{equation}
  \sum_{i\in O} |\bt_i|_1=s+\sum_{i\in E} |\bt_i|_1
  \end{equation}
  It follows that there exist $j_0,j_1\in[r]$
  satisfying
  \begin{equation}
  \sum_{i\in O} \bt_i(j_0)\geq \sum_{i\in E} \bt_i(j_0) \;\;\text{ and }\;\;
  \sum_{i\in O} \bt_i(j_1)\leq \sum_{i\in E} \bt_i(j_1).
  \end{equation}
   We note that
  for the existence of~$j_1$ we
  use~$s<r$. Consequently,~$f(\bt_1,\dots,\bt_m)$ contains at least
  a~$0$ and a~$1$ (in coordinates~$j_0$ and~$j_1$
  respectively). Then~$f(\bt_1,\dots,\bt_m)\in R^{\bT}$ as desired.
   
  It then follows from (\cite{BrakensiekG18}, Section 3.2) that any
  instance~$\bI$ of~$\PCSP(\bS,\bT)$ is homomorphic to~$\bT$ whenever
  the following {\em augmented~$\SA^1$ feasibility} condition holds:
  for each variable~$v\in I$ there exists a~$b_v \in \{0,1\}$ such
  that there is a feasible solution of~$\SA^1(\bI,\bS)$ that
  sets~$x_{v \mapsto b_v}$ to~$1$. We are left to show that
  if~$\SA^2(\bI,\bS)$ has a feasible solution, then the
  augmented~$\SA^1$ feasibility condition holds.

  To avoid confusion we use~$x'_f$ to refer to the variables
  in~$\SA^2(\bI,\bS)$ and~$x_f$ to refer to the variables
  in~$\SA^1(\bI,\bS)$.  Fix~$v \in I$.  It follows from the first three
  types of constraint in the definition of~$\SA^2$ that there exists
  some~$b \in \{0,1\}$ that~$x'_{v \mapsto b} > 0$;
  let~$d := x'_{v \mapsto b}$. Now construct a solution
  of~$\SA^1(\bI,\bS)$ as follows: for every~$u\in I$ and~$a \in B$,
  set~$x_{u \mapsto a}$ to~$x'_{u \mapsto a, v \mapsto
    b}/d$. By choice of~$d$ we have~$x_{v \mapsto b}=1$ as
  desired. It is easy to verify that the assignment thus defined is a
  feasible solution of~$\SA^1(\bI,\bS)$.
\end{proof}

Hence, from Theorem~\ref{the:main} and Lemma \ref{le:ExNAE} we have:

\begin{corollary} \label{cor:separation} 
  There is a PCSP template~$(\bS,\bT)$ that has
  SA-width~$2$, and hence bounded SA-width, but does not have
  sublinear or bounded width. Concretely,
  setting~$\bS = \EXACTLY{s}{r}$ and~$\bT = \NAE{r}$ where~$s$ and~$r$
  are integers such that~$0 < s < r$ and~$r > 2$ gives such an
  example.
\end{corollary}

\noindent It should be pointed out that the condition that~$r > 2$ in
Corollary~\ref{cor:separation} is necessary. Indeed, the
template~$\EXACTLY{1}{2}$ is isomorphic to~$\clique_2$, and we already
know that~$\clique_2$, and hence~$(\EXACTLY{1}{2},\NAE{2})$, has width
three (see Example~\ref{ex:majority}). It is also easy to see that,
consistently with this,~$(\EXACTLY{1}{2},\NAE{2})$ satisfies the
condition of Theorem~\ref{the:main}: the composition of the binary
relation of~$\EXACTLY{1}{2}$ with itself is the equality relation
on~$\{0,1\}$, which is not the full binary relation~$\{0,1\}^2$.

\section{Proof of Theorem~\ref{the:main}} \label{sec:proof}

The proof is structured in five parts. The first part sets the
stage. In the second part we define a probability distribution on
instances and prove that a random instance~$\bI$ sampled from this
distribution has, with high probability, two key properties: (a) it is
\emph{dense enough} to guarantee that~$\bI \not\to \bT$, and (b) it is
also \emph{sparse enough}, in a sense compatible with (a), to
guarantee that~$\bI \leq_k \bS$, as will be proved in the next
part. In the third part we show that any instance~$\bI$ that satisfies
the sparsity condition indeed satisfies~$\bI \leq_k \bS$. In the
fourth part we discuss the setting of parameters that satisfies all
the required conditions to prove the theorem. Finally, the fifth part
of the proof focuses on the special case of the theorem that applies
to digraphs (structures with a single binary relation). For this
special case we are able to slightly improve the parameters; this
special case for (di)graphs will be used in
Section~\ref{sec:approximatecoloring}.

\subsection{Setting the stage} \label{sec:stage}

Let~$\sig$ be a fixed signature and let~$(\bS,\bT)$ be a PCSP template
of signature~$\sig$ that has sublinear width.
Let~$\bS' := (S,\prod_{R\in\sig} R^{\bS})$
and~$\bT' := (T,\prod_{R\in\sig} R^{\bT})$ and note that~$\bS \to \bT$
implies~$\bS' \to \bT'$. Since the signature~$\sig$ is finite and
fixed, the template~$(\bS,\bT)$ has sublinear width if and only if the
template~$(\bS',\bT')$ has sublinear width.  Also, it holds
that~$\graph(\bS) = \graph(\bS')$, and that~$\bT$ is reflexive if and
only if~$\bT'$ is reflexive. Therefore, it suffices to prove the
theorem when the signature~$\sigma$ consists of a single relation
symbol~$R$. Furthermore, since any unary relation is either reflexive
or empty, we may assume that the arity~$r$ of~$R$ is at least~$2$;
i.e.,~$r \geq 2$.

Let~$(\bS,\bT)$ satisfy the following assumptions:

\begin{center}
\begin{tabular}{lll}
A1: & & for all $U,V \in \graph(\bS)$ we have $U \circ V = S^2$, \\
A2: & & for all~$a \in T$, the reflexive tuple~$(a,a,\ldots,a)$ is not in~$R^{\bT}$.
\end{tabular}
\end{center}

Let~$p = |S|$ and~$q = |T|$ be the cardinalities of the domains
of~$\bS$ and~$\bT$, respectively, and let~$k = k(n)$ be an integer
function such that~$k(n) = o(n)$. Our goal is to show that there exist
arbitrarily large instances~$\bI$ that witness that the PCSP
template~$(\bS,\bT)$ does not have width~$k$; i.e., the instance~$\bI$
is such that~$\bI \leq_k \bS$ and~$\bI \not\rightarrow \bT$. We show
that such an instance~$\bI$ exists by the probabilistic method.

In anticipation for the proof, in addition to the data~$r,k,p,q$, we
fix some real parameters~$\delta,\beta,\alpha,c,d$, as well as an
integer parameter~$n$. These parameters are required to satisfy the
following conditions:

\begin{center}
\begin{tabular}{lll} 
  \text{C1:} & &~$0 < \delta \leq 1/((r+1)(3r+1))$, \\
  \text{C2:} & &~$0 < \beta \leq (1+\delta)/(r-1)$, \\
  \text{C3:} & &~$0 < \alpha \leq (\beta/d)^{1/(r-1)} (r/e)^{r/(r-1)}$, \\
  \text{C4:} & &~$c \geq kp/\delta$, \\
  \text{C5:} & &~$n \geq \max\{c/(\alpha\beta),q\}$, \\
  \text{C6:} & &~$1 \leq d \leq n^{r-1}$, \\
             & & \\
  \text{C7:} & &~$p_1(r,d,n,q) + p_2(r,d,n,\alpha,\beta) < 1$,
\end{tabular}
\end{center}
where
\begin{align}
& p_1(r,d,n,q) := q^n \exp{\left({-{dn}/{(r^r q^{r-1}}}\right)}, \\
& p_2(r,d,n,\alpha,\beta) :=
\textstyle{\sum_{v=1}^{\lfloor{\alpha n}\rfloor} \left({(n/v)^{1-(r-1)\beta} d^{\beta}
e^{1+(r+1)\beta} r^{-r\beta} \beta^{-\beta}}\right)^v}.
\end{align}
We have separated the first six conditions C1--C6 from the last one~C7
because the first six are easily feasible by themselves; fulfilling
Condition~C7 simultaneously is more delicate.  At the end of the proof
we discuss a settings of the parameters~$\delta,\beta,\alpha,c,d$
and~$n$ that satisfy Conditions C1--C7. For now, we assume that the
conditions are feasible.

\subsection{Probabilistic construction}

Let~$H$ denote a random Erd\H{o}s-R\'enyi~$r$-uniform hypergraph
with~$n$ vertices and edge probability~$d/n^{r-1}$;
i.e.,~$V(H) = [n]$, and each~$r$-element subset~$C \subseteq [n]$ is
or is not an edge in~$E(H)$, independently, with
probability~$d/n^{r-1}$. Note that~$d/n^{r-1}$ is a proper probability
by Condition~C6. Let~$\bI = \bI(H)$ be the (random) instance with
domain~$[n]$ that has one~$r$-tuple~$\bv_C$ in~$R^{\bI}$ for
each~$C \in E(H)$, where~$\bv_C$ is obtained by ordering the elements
in~$C$ is some arbitrary way.

\begin{lemma} \label{le:nohomomorphic}
  The probability that~$\bI$ is homomorphic to~$\bT$ is at
  most~$p_1(r,d,n,q)$.
\end{lemma}

\begin{proof}
  Let~$h$ be any mapping from the domain~$I$ of~$\bI$ to the
  domain~$T$ of~$\bT$. Let~$S_1,\ldots,S_q$ be the partition induced
  by~$h$; i.e.,~$S_i = h^{-1}(i)$ for~$i = 1,\ldots,q$.
  Let~$n_i := |S_i|$. By Assumption~A2, the relation~$R^{\bT}$ does
  not contain a reflexive tuple. Therefore, for all~$i \in [q]$
  and~$C \subseteq S_i$ we have~$h(\bv_C)\not\in R^{\bT}$. By
  independence, the probability that~$h(\bv_C)$ belongs to~$R^{\bT}$
  for every~$\bv_C\in R^{\bI}$ is at most~$(1-d/n^{r-1})^N$,
  where~$N = \sum_{i=1}^q \binom{n_i}{r}$.  Now note
  that~$\binom{n_i}{r} \geq (n_i/r)^r$, and that, subject to the
  constraints~$x_1,\ldots,x_q \geq 0$ and~$\sum_{i=1}^q x_i = n$, the
  sum~$\sum_{i=1}^q (x_i/r)^r$ is minimized
  at~$x_1 = \cdots = x_q = n/q$. It follows
  that~$N \geq q (n/(rq))^r = n^r/(r^r q^{r-1})$.
  Therefore,~$(1-d/n^{r-1})^N \leq \exp(-dn^r/(n^{r-1} r^r q^{r-1}))
  \leq \exp(-dn/(r^r q^{r-1}))$.  By the union bound the probability
  that there is a homomorphism from~$\bI$ to~$\bT$ is at
  most~$q^n \exp(-dn/(r^r q^{r-1})$, which equals~$p_1(r,d,n,q)$.
\end{proof}

We say that~$\bI$ is~\emph{$(\alpha,\beta)$-sparse} if every
substructure~$\bJ$ of~$\bI$ with~$v$ elements
satisfying~$1 \leq v \leq \alpha n$ has less than~$\beta v$ many
tuples. Equivalently,~$\bI$ is~$(\alpha,\beta)$-sparse if every
substructure~$\bJ$ of~$\bI$ with at most~$\alpha n$ elements and at
least~$\beta v$ many tuples has more than~$v$ elements.  The argument
for the lemma below is almost the same as the one for Lemma~1
in~\cite{ChvatalSzemeredi1988} except that we need to apply it to our
probability model. Also as in~\cite{ChvatalSzemeredi1988}, we use the
following Chernoff-like bound
\begin{equation}
\sum_{i=\lceil{tm}\rceil}^m {m \choose i} \gamma^i (1-\gamma)^{m-i} \leq
\left({e\gamma/t}\right)^{tm}, \label{eqn:chernofflike}
\end{equation}
which holds for every integer~$m \geq 1$, every
real~$\gamma \in [0,1]$, and every real~$t \in (\gamma,1]$.

\begin{lemma} \label{le:lowdensity}
  The probability that~$\bI$ is not~$(\alpha,\beta)$-sparse is at
  most~$p_2(r,d,n,\alpha,\beta)$.
\end{lemma}

\begin{proof}
For each integer~$v$ such that~$1 \leq v \leq \alpha n$,
set~$m_v = {v \choose r}$. Setting~$p_0 = d/n^{r-1}$, the probability
that~$\bI$ is not~$(\alpha,\beta)$-sparse is bounded by
\begin{equation}
 \sum_{v = r}^{\lfloor{\alpha n}\rfloor} {n \choose v}\sum_{i =
  \lceil{\beta v}\rceil}^{m_v} {m_v \choose i} p_0^i (1-p_0)^{m_v-i}.
\label{eqn:prob}
\end{equation}
Now set~$t_v = \beta v/m_v$. Using~${v \choose r} < (ve/r)^r$ note
that~$t_v > \beta r^r/(e^r v^{r-1}) \geq d/n^{r-1} = p_0$
for~$v \leq \alpha n$ since~$\alpha$ satisfies Condition~C3.
If~$t_v > 1$, then the inner sum in~\eqref{eqn:prob} is zero, while
if~$t_v \leq 1$, then we have~$t_v \in (p_0,1]$,
so~\eqref{eqn:chernofflike} applies to bound~\eqref{eqn:prob} by
\begin{equation}
  \sum_{v = r}^{\lfloor{\alpha n}\rfloor} {n
    \choose v} \left({ep_0/t_v}\right)^{t_v m_v}.  
\end{equation}
Using~${n \choose v} \leq (ne/v)^v$ and~${v \choose r} \leq (ve/r)^r$,
we bound this further by
\begin{equation}
  \sum_{v = r}^{\lfloor{\alpha n}\rfloor}
  \left({(n/v)^{1-(r-1)\beta} d^\beta e^{1+(r+1)\beta} r^{-r\beta}\beta^{-\beta}}\right)^v,
\end{equation}
which is bounded by $p_2(r,d,n,\alpha,\beta)$.
\end{proof}

\begin{lemma}
  \label{le:largev}
  There exists a structure~$\bI$ with~$n$ elements that is not
  homomorphic to~$\bT$ and is~$(\alpha,\beta)$-sparse. In
  particular,~$\bI$ is such that for every substructure~$\bJ$ of~$\bI$
  and every integer~$m \geq 0$, if~$\bJ$ has~$m$ many tuples
  and~$m \leq c$, then~$\bJ$ has more than~$(r-1)m/(1+\delta)$
  elements.
\end{lemma}

\begin{proof}
  The existence of~$\bI$ follows directly from the assumption that
  Condition~C7 holds and Lemmas~\ref{le:nohomomorphic}
  and~\ref{le:lowdensity}. For the second part of the statement,
  assume that~$\bJ$ is a substructure of~$\bI$ that has~$v$ many
  elements and~$m$ many tuples, with~$m \leq c$
  and~$v \leq (r-1)m/(1+\delta)$.  In
  particular~$v \leq m/\beta \leq c/\beta \leq \alpha n$ since~$\beta$
  satisfies Condition~C2 and~$n$ satisfies Condition~C5. But~$\bI$
  is~$(\alpha,\beta)$-sparse and it follows that~$\bJ$ has less
  than~$\beta v \leq m$ many tuples; a contradiction.
\end{proof}

\subsection{Proof of consistency}

Still assuming that the setting of parameters satisfies the
conditions~C1--C7, let~$\bI$ be a structure as in the second part of
Lemma~\ref{le:largev}.  We shall now prove that~$\bI \leq_k \bS$.

The proof adapts the notion of \emph{boundary} from
\cite{MolloySalavatipour2007}; this was used to prove size lower
bounds for the Resolution proof complexity of random CSPs. Let~$\bJ$
be any substructure of~$\bI$. A set~$D \subseteq J$ is said to be a
\emph{boundary set of~$\bJ$} if it is non-empty and every homomorphism
from~$\bJ|_{J\setminus D}$ to~$\bS$ extends to a homomorphism
from~$\bJ$ to~$\bS$. We introduce two types of subsets~$D \subseteq J$
of~$\bJ$ and show that any set of any of these types is a boundary set
of~$\bJ$.

We define the {\em degree} of an element in~$\bJ$ to be the number of
tuples of~$R^{\bJ}$ in which it appears.  We say that~$D \subseteq J$
is of type~(1) if~$D=\{x_1,\dots,x_{r-1}\}$, where~$x_1,\dots,x_{r-1}$
are all distinct and have degree one in~$\bJ$, and there
exists~$\bv_C\in R^{\bJ}$ with~$D\subseteq C$. In this case we say
that the set~$C$ {\em witnesses} that~$D$ is of type~(1).  We say
that~$D \subseteq J$ is of type~(2)
if~$D=\{x_1,y_1,\dots,x_{r-2},y_{r-2},z\}$,
where~$x_1,y_1,\dots,x_{r-2},y_{r-2}$ are all distinct and have degree
one in~$\bJ$, also~$z$ is distinct from the rest of elements in~$D$
and has degree two in~$\bJ$, and there exist two different
tuples~$\bv_{C_1},\bv_{C_2}\in R^{\bJ}$ such
that~$\{x_1,\dots,x_{r-2},z\}\subseteq C_1$
and~$\{y_1,\dots,y_{r-2},z\}\subseteq C_2$. In this case we say that
the sets~$C_1,C_2$ {\em witness} that~$D$ is of type~(2). We check
below that, since~$\graph(R^{\bS})$ satisfies Assumption~A1, every set
of these two types is a boundary set.

\begin{lemma}
  For every substructure~$\bJ$ of~$\bI$ and every~$D \subseteq J$,
  if~$D$ is of type (1) or (2) in~$\bJ$, then~$D$ is a boundary set
  of~$\bJ$.
\end{lemma}
\begin{proof}
  Fix~$\bJ$ and $D$ as in the hypothesis and let~$h$ be a
  homomomorphism from~$\bJ|_{J\setminus D}$ to~$\bS$.

  Firstly, assume that~$D=\{x_1,\dots,x_{r-1}\}$ is of type~(1), and
  let~$C$ witness so; i.e.,~$\bv_C \in R^{\bJ}$ and~$D \subseteq C$.
  Let~$i_0\in [r]$ be such that~$\bv_C(i_0) \in J\setminus D$. By
  Assumption~A1 on~$\script{G}(R^{\bS})$ we have that
  \begin{equation}
  \pr_{i_0} R^{\bS}=S.
  \end{equation} 
  Consequently, there is a tuple~$\ba \in R^{\bS}$ such
  that~$\ba(i_0) = h(\bv_C(i_0))$. We extend~$h$ to map every element
  of~$J$ by setting~$h(\bv_C(i)) = \ba(i)$ for
  every~$i \in [r]\setminus\{i_0\}$. Since each~$x \in D$ appears
  in~$\bv_C$ but in no other tuple of~$\bJ$, the result is a
  homomorphism from~$\bJ$ to~$\bS$.

  Secondly, assume that~$D=\{x_1,y_1,\dots,x_{r-2},y_{r-2},z\}$ is of
  type~(2), and let~$C_1,C_2$ witness so;
  i.e.,~$\bv_{C_1},\bv_{C_2}\in R^{\bJ}$ are distinct,
  and~$\{x_1,\dots,x_{r-2},z\}\subseteq C_1$
  and~$\{y_1,\dots,y_{r-2},z\}\subseteq C_2$. For~$j=1,2$,
  let~$i_j\in [r]$ be such that~$\bv_{C_j}(i_j)\in J \setminus D$ and
  let~$k_j\in [r] \setminus \{i_j\}$ be such that~$\bv_{C_j}(k_j)=z$.
  Again, by Assumption~A1 on~$\script{G}(R^{\bS})$, we
  have that
  \begin{equation}
  \pr_{i_1,k_1} R^{\bS}\circ \pr_{k_2,i_2} R^{\bS}=S^2.
  \end{equation} 
  Consequently, there exist
  tuples~$\ba_1,\ba_2 \in R^{\bS}$ such
  that~$\ba_j(i_j)=h(\bv_{C_j}(i_j))$ for~$j=1,2$
  and~$\ba_1(k_1)=\ba_2(k_2)$. We extend~$h$ to map every element
  of~$J$ by setting~$h(\bv_{C_j}(i)) = \ba_j(i)$ for~$j=1,2$ and
  every~$i\in [r]\setminus i_j$. Since each~$x \in D$ appears
  in~$\bv_{C_1}$ or in~$\bv_{C_2}$ (or in both, in case of~$z$) but in
  no other tuple of~$\bJ$, the result is a homomorphism from~$\bJ$
  to~$\bS$.
\end{proof}

\begin{lemma}
  \label{le:lotofboundaries}
  For every substructure~$\bJ$ of~$\bI$ and every integer~$m \geq 0$,
  if~$\bJ$ has~$m$ many tuples and~$m \leq c$, then~$\bJ$ has at
  least~$\delta m$ many pairwise disjoint boundary sets.
\end{lemma}

\begin{proof}
  We can assume that~$\bJ$ does not contain elements of degree zero
  since every boundary set of any substructure obtained by removing
  elements of degree~$0$ from~$\bJ$ is also a boundary set of~$\bJ$.
  We shall show that the total number of boundary sets in~$\bJ$ that
  are of type~(1) or~(2) is at least~$\delta m$. Since any two
  distinct boundary sets of these types must be disjoint, the claim
  will follow.

  Let~$\script{C} = \{ C \mid \bv_C\in R^{\bJ}\}$.
  For~$C=\{x_1,\dots,x_r\} \in \script{C}$, define 
  \begin{equation}
  \invd(C) := \textstyle{\sum_{i=1}^r 1/d_i}
\end{equation} 
where~$d_i$ is the degree of~$x_i$ in~$\bJ$ (\emph{sdr} stands for sum
of degree reciprocals). Since~$\bJ$ does not have elements of degree
zero, it is easy to see that~$\sum_{C\in \script{C}} \invd(C)$ is
equal to the number~$v$ of elements in~$\bJ$. The idea of the proof is
the following. Since, by Lemma~\ref{le:largev} we have
that~$v\geq (r-1)m/(1+\delta)$, there is a large number of
sets~$C\in\script{C}$
with~$\invd(C)\geq (r-1)/(1+\delta)$. Since~$\delta$ is small
enough each one of these sets must have at least~$r-2$ vertices of
degree one, and one vertex of degree at most two. From this large pool
of sets it is not difficult to find a large number of witnesses for
boundary sets of type~(1) or~(2). We formalize this below.

Let~$\script{D}$ be a collection of boundary sets of types~(1) and~(2)
with the largest possible cardinality and assume towards a
contradiction that~$|\script{D}|<\delta m$.  We partition~$\script{C}$
in three sets~$\script{C}_1$,~$\script{C}_2$,~$\script{C}_3$.  The
set~$\script{C}_1$ contains the witness~$C \in \script{C}$ of every
boundary set of type~(1) in~$\script{D}$, and exactly one among the
two sets~$C_1,C_2 \in \script{C}$ that witness some boundary set of
type~(2) in~$\script{D}$. Note
that~$|\script{C}_1| = |\script{D}| < \delta m$. The
set~$\script{C}_2$ contains all~$C\in\script{C}\setminus\script{C}_1$
such that~$\invd(C) > r-4/3$. The set~$\script{C}_3$ contains the
rest, i.e.,
all~$C \in \script{C} \setminus (\script{C}_1 \cup \script{C}_2)$.
Note that every~$C\in\script{C}\setminus\script{C}_1$ must contain at
least two elements of degree larger than one: otherwise it would be
the witness of a boundary set of type (1), against the maximality
of~$\script{D}$. In particular, for
any~$C\in\script{C}\setminus\script{C}_1$, we have~$\invd(C)\leq r-1$.
Note also that~$\invd(C) \leq r$ for every~$C \in \script{C}$,
and~$\invd(C) \leq r-4/3$ for every~$C \in \script{C}_3$.

Let~$\alpha_0$ and~$\beta_0$ be reals in~$[0,1]$ such
that~$|\script{C}_1|=\alpha_0
m$,~$|\script{C}_2|=(1-\alpha_0-\beta_0)m$,
and~$|\script{C}_3|=\beta_0 m$. Since each boundary set
in~$\script{D}$ contributes exactly one set to~$\script{C}_1$
and~$|\script{D}| < \delta m$, we have~$\alpha_0 < \delta$. We shall
prove that~$\beta_0 \leq 3r\delta$. Assume otherwise;
i.e.,~$\beta_0 > 3r\delta$. Recall
that~$v = \sum_{C \in \script{C}} \invd(C)$ and, hence,
\begin{align}
v & \leq |\script{C}_1|r+|\script{C}_2|(r-1)+|\script{C}_3|(r-4/3) \label{eqn:seqfirst} \\
& = \alpha_0 mr + (1-\alpha_0-\beta_0)m (r-1) + \beta_0 m (r-4/3)  \\
&= m (r-1+\alpha_0-\beta_0/3) \\
&< m (r-1)(1-\delta), \label{eqn:seqlast}
\end{align}
where the first inequality follows from the fact
that~$\script{C}_1,\script{C}_2,\script{C}_3$ is a partition
of~$\script{C}$ and the already noted bounds on~$\invd(C)$ for the~$C$
in these sets, the first equality follows from the choices
of~$\alpha_0$ and~$\beta_0$, the second equality follows from plain
arithmetic, and the strict inequality follows from~$\alpha_0 < \delta$
and the assumption that~$\beta_0 > 3r\delta$.  On the other hand, by
Lemma~\ref{le:largev} we have that~$v \geq
(r-1)m/(1+\delta)$. Combined
with~\eqref{eqn:seqfirst}-\eqref{eqn:seqlast}, this means
that~$1-\delta\geq 1/(1+\delta)$, which is impossible since~$\delta>0$
by Condition~C1.

It follows that~$|\script{C}_2|\geq (1-(1+3r)\delta) m$.
Let~$C\in\script{C}_2$. Since~$r - 3/2 < r - 4/3<\invd(C)$, the
set~$C$ must contain at least~$r-2$ elements of degree one. We also
know that the remaining two elements in~$C$ must have degree at least
two since, otherwise,~$C$ would witness a boundary set of type~(1),
against the maximality of~$\script{D}$. Again from~$r-4/3<\invd(C)$ it
follows that at least one of the two remaining elements must have
degree exactly two. Two sets in~$\script{C}_2$ cannot share an element
of degree two since otherwise they would witness a boundary set of
type~(2), against the maximality of~$\script{D}$ again. Consequently,
each one of the~$(1-(1+3r)\delta)m$ sets in~$\script{C}_2$ contains an
element of degree two which must also appear in some set
of~$\script{C}_1\cup \script{C}_3$.
Since~$|\script{C}_1\cup\script{C}_3|=(\alpha_0+\beta_0)m\leq
(1+3r)\delta m$, the total number of elements that can appear
in~$\script{C}_1\cup\script{C}_3$ is at most~$r(1+3r)\delta m$, which
yields~$1-(1+3r)\delta<r(1+3r)\delta$, against Condition~C1.
\end{proof}

\begin{lemma}
\label{co:csatisfiable}
For every substructure~$\bJ$ of~$\bI$, if~$\bJ$ has at most~$c$ many
tuples, then~$\bJ$ is homomorphic to~$\bS$.
\end{lemma}

\begin{proof}
  Assume, towards a contradiction, that there is a substructure~$\bJ$
  of~$\bI$ with~$m \leq c$ many tuples that is not homomorphic
  to~$\bS$. We can further assume that~$m \geq 1$, and that~$\bJ$ is
  minimal in the sense that any of its proper substructures is
  homomorphic to~$\bS$. By Lemma~\ref{le:lotofboundaries}, the
  substructure~$\bJ$ has at least~$\delta m>0$ boundary sets.
  Let~$D \subseteq J$ be any boundary set of~$\bJ$, which is non-empty
  by definition. By the minimality of~$\bJ$, there is an
  homomorphism~$f$ from~$\bJ|_{J\setminus D}$ to~$\bS$. By the
  definition of boundary set,~$f$ can be extended to a homomorphism
  from~$\bJ$ to~$\bS$; a contradiction.
\end{proof}

Let~$\bJ$ be a substructure of~$\bI$. A mapping~$h:X\rightarrow S$
with~$X\subseteq I$ is said to be \emph{consistent with~$\bJ$} if
there is an homomorphism~$g$ from~$\bJ$ to~$\bS$ such that~$h$ and~$g$
agree on the intersection~$\dom(h)\cap\dom(g)$. For the next lemma,
recall that~$p = |S|$.

\begin{lemma}
  \label{le:double}
For every partial homomorphism~$h$ from~$\bI$ to~$\bS$
with~$\dom(h) \leq k$, if~$h$ is consistent with every substructure
of~$\bI$ with a most~$c/p$ many tuples, then~$h$ is also consistent
with every substructure of~$\bI$ with at most~$c$ many tuples.
\end{lemma}
\begin{proof}
  Fix~$h$ as in the hypothesis and assume, towards a contradiction,
  that~$h$ is not consistent with some substructure~$\bJ$ of~$\bI$
  with~$m \leq c$ many tuples. We can further assume that~$\bJ$ is
  minimal in the sense that~$h$ is consistent with any proper
  substructure of~$\bJ$. By assumption we have~$c/p < m \leq c$,
  where~$p = |S|$. It follows from Lemma~\ref{le:lotofboundaries}
  that~$\bJ$ has a collection~$\script{D}$ of at
  least~$\delta m > \delta c/p$ many pairwise disjoint boundary
  sets. By the minimality of~$\bJ$, for each boundary
  set~$D \in \script{D}$ of~$\bJ$, which is non-empty, we have
  that~$h$ is consistent with the substructure~$\bJ|_{J \setminus
    D}$. Therefore, there exists a homomorphism~$g_D$
  from~$\bJ|_{J \setminus D}$ to~$\bS$ that agrees with~$h$. Since~$D$
  is a boundary set of~$\bJ$, the homomorphism~$g_D$ extends to a
  homomorphism~$h_D$ from~$\bJ$ to~$\bS$. But~$h$ is not consistent
  with~$\bJ$ which means that~$h_D$ and~$h$ disagree somewhere in~$D$;
  in particular,~$D \cap \dom(h) \not= \emptyset$. Since the boundary
  sets in~$\script{D}$ are pairwise disjoint, we
  get~$|\dom(h)| > \delta c/p$, which is at least~$k$
  by Condition~C4; a contradiction.
\end{proof}

\begin{lemma} \label{lem:consistent}
  There is a $k$-strategy on $\bI$ and $\bS$; i.e., $\bI \leq_k \bS$.
\end{lemma}
\begin{proof}
  Define~$\script{H}$ to be set of all partial mappings from~$I$
  to~$S$ with domain size at most~$k$ that are consistent with every
  substructure~$\bJ$ of~$\bI$ with at most~$c$ many tuples. We show
  that~$\script{H}$ is a~$k$-strategy. 

  First, note that every~$h\in \script{H}$ is a partial homomorphism
  from~$\bI$ to~$\bS$. To see this, take any tuple~$\bv_C$
  in~$\bI|_{\dom(h)}$ and we show that~$h(\bv_C)$ is
  in~$R^{\bS}$. Let~$\bJ$ be the substructure of~$\bI$ that contains
  only~$\bv_C$. Since~$\bJ$ has only one tuple and~$c \geq 1$, by the
  definition of~$\script{H}$ and the fact that~$h$ is in~$\script{H}$
  we have that~$h$ is consistent with~$\bJ$; that is, there is a
  homomorphism~$g$ from~$\bJ$ to~$\bS$ that agrees with~$h$ on the
  intersection~$C = \dom(h) \cap \dom(g)$;
  i.e.,~$h(\bv_C) = g(\bv_C) \in R^{\bS}$.
  Also it follows directly from Lemma~\ref{co:csatisfiable}
  that~$\script{H}$ is non-empty as it contains the partial
  mapping~$h$ with~$\dom(h)=\emptyset$.

  It remains to be shown that~$\script{H}$ is closed under
  restrictions and satisfies the extension property up to~$k$. The
  closure under restrictions follows directly from the
  definitions. Let us then verify that~$\script{H}$ satisfies the
  extension property up to~$k$. Fix~$h\in \script{H}$
  with~$|\dom(h)|<k$, and fix~$x \in I\setminus\dom(h)$. For
  every~$a\in S$, let~$h_a$ be the extension of~$h$ with
  domain~$\dom(h)\cup\{x\}$ that maps~$x$ to~$a$. We claim that there
  exists some~$a\in S$ such that~$h_a$ is consistent with all
  substructures of~$\bI$ with at most~$c/p$ many tuples. Once the
  claim is proved we can conclude from Lemma~\ref{le:double}
  that~$h_a$ belongs to~$\script{H}$ and the proof will be
  complete. To prove the claim, assume that for each~$a\in S$ there is
  some substructure~$\bJ_a$ of~$\bI$ that falsifies it. Then,~$h$ is
  not consistent with~$\bigcup_{a\in S} \bJ_a$, which is a
  substructure of~$\bI$ that has at most~$c$ many tuples,
  since~$|J_a| \leq c/p$ and~$p = |S|$; a contradiction.
\end{proof}

\subsection{Settings of parameters}

The data~$r,p,q$ are fixed and independent of~$n$, but~$k$ is set
to~$\epsilon n$ for a small enough positive constant~$\epsilon > 0$ to
be determined next (in Equations~\eqref{eqn:epsilon1}
and~\eqref{eqn:epsilon2} below). We set~$\delta$ and~$\beta$ to their
upper bounds in Conditions~C1 and~C2. In particular,~$\delta$
and~$\beta$ are constants independent of~$n$
and~$\delta = (r-1)\beta - 1$.  Set~$d = r^r q^{r-1}\ln(2q)$, so~$d$
is also a constant independent of~$n$.
Set~$\alpha = \epsilon p/(\delta\beta)$.  If~$\epsilon$ is small
enough, namely, if
\begin{equation}
0 < \epsilon < (\delta\beta/p)(\beta/d)^{1/(r-1)} (r/e)^{r/(r-1)},
\label{eqn:epsilon1}
\end{equation}
then Condition~C3 is satisfied.  Set~$c = kp/\delta$, so that
Condition~C4 is satisfied. Now choose~$n$ is large enough so
that~$n \geq q$; observe that the choices of~$\alpha$ and~$c$ are made
so that~$c/(\alpha\beta) = n$, so Condition~C5 is satisfied. Since~$d$
is a constant independent of~$n$, Condition~C6 is also satisfied,
if~$n$ is large enough. We still need to check that Condition~C7
holds. By the choice of~$d$ we have that~$p_1(r,d,n,q) = 1/2^n$, which
approaches~$0$ as~$n$ grows to infinity.  Set
\begin{align*}
  & \rho_1(n) := (1/\sqrt{n})^{\delta} d^{\beta} e^{1+(r+1)\beta}
    r^{-r\beta} \beta^{-\beta}, \\
  & \rho_2(n) := (\epsilon p/(\delta\beta))^\delta d^{\beta} e^{1+(r+1)\beta} r^{-r\beta} \beta^{-\beta}.
\end{align*}
Splitting the sum that defines~$p_2(r,d,n,\alpha,\beta)$
into~$v \leq \lfloor{\sqrt{n}}\rfloor$
and~$v \geq \lfloor{\sqrt{n}}\rfloor+1$ we get
\begin{align}
p_2(r,d,n,\alpha,\beta) \leq
\textstyle{\sum_{v=1}^{\lfloor{\sqrt{n}}\rfloor} \rho_1(n)^v} 
+
\textstyle{\sum_{v=\lfloor{\sqrt{n}}\rfloor+1}^{\lfloor{\alpha n}\rfloor}
\rho_2(n)^v}.
\label{eqn:firstandsecond}
\end{align}
Now note that~$\rho_1(n) < 1/3$ for large enough~$n$
because~$\delta,d,r,\beta,p$ are constants independent of~$n$. Also,
if~$\epsilon$ is positive but small enough, namely, if
\begin{equation}
0 < \epsilon < (\delta\beta/(3p)) 
d^{-\beta/\delta} e^{-1/\delta-(r+1)\beta/\delta} r^{r\beta/\delta} \beta^{\beta/\delta},
\label{eqn:epsilon2}
\end{equation} 
then also~$\rho_2(n) < 1/3$ again because~$\delta,d,r,\beta,p$ are
constants independent of~$n$ and of~$\epsilon$. Thus, the sum
in~\eqref{eqn:firstandsecond} is strictly less
than~$\sum_{v=1}^\infty (1/3)^v = 1/2$.
Therefore,~$p_1 + p_2 < 1/2^n + 1/2 < 1$ for large enough~$n$, which
proves that Condition~C7 is satisfied, as was to be shown.

\subsection{Special case of digraphs} \label{sec:digraphs}

In this section we assume that~$r = 2$ and use it to improve the
parameters. The improvement is that we can replace the~$1/21$ that
would result from plugging~$r = 2$ into the right-hand side of
Condition~C1 by~$1/2$. The price to pay for this is that the
right-hand side of Condition~C4 becomes slightly bigger.

We indicate the required changes in the previous
proof. Conditions~C1 and C4 are replaced by the following:
\begin{center}
\begin{tabular}{lll} 
\text{C1':} & & $0 < \delta < 1/2$, \\
\text{C4':} & & $c \geq kp/\delta'$. \\
\end{tabular}
\end{center}
where
\begin{equation}
\delta' := (1-2\delta)/(6(1+\delta)) \label{eqn:defdelta'}
\end{equation}
Lemma~\ref{le:lotofboundaries} becomes the following:

\begin{lemma}
  \label{le:lotofboundaries2}
  For every substructure~$\bJ$ of~$\bI$ and every integer~$m \geq 0$,
  if~$r=2$ and~$\bJ$ has~$m$ many tuples and~$m \leq c$, then~$\bJ$
  has at least~$\delta' m$ many pairwise disjoint
  boundary sets.
\end{lemma}

\begin{proof}
  We can assume that~$\bJ$ does not contain elements of degree zero
  since every boundary set of any substructure obtained by removing
  elements of degree~$0$ from~$\bJ$ is also a boundary set of~$\bJ$.
  In the special case~$r = 2$, the boundary sets of types~(1) and~(2)
  are in one-to-one correspondence with the vertices of degree one and
  the vertices of degree two, respectively. Thus, since~$\bJ$ does not
  contain elements of degree zero, and since each boundary set
  involves two vertices, it suffices to show that~$\bJ$ contains at
  least~$2\delta' m$ many vertices of degree at
  most two; this will give at least~$\delta' m$
  many pairwise disjoint boundary sets.

  Let~$v = |J|$ and let~$X$ denote the random variable that equals the
  degree, in~$\bJ$, of a uniformly chosen random element of~$J$. The
  sum of the degrees of the elements in~$\bJ$
  is~$2m$. Therefore,~$\mathbb{E}[X] = 2m/v$. Now recall that, by
  Lemma~\ref{le:largev}, we
  have~$v \geq (r-1)m/(1+\delta) = m/(1+\delta)$, since~$m \leq c$
  and~$r = 2$. By Markov's inequality, the probability that~$X \geq 3$
  is bounded by~$\mathbb{E}[X]/3 \leq (2m/v)/3 \leq
  2(1+\delta)/3$. Therefore, at least a~$1-2(1+\delta)/3$ fraction of
  the elements of~$\bJ$ have degree strictly less than~$3$ in~$\bJ$,
  which means that~$\bJ$ has at
  least~$(1-2(1+\delta)/3)v \geq (1-2(1+\delta)/3)m/(1+\delta) =
  2\delta' m$ elements of degree at most two. The
  lemma is proved.
\end{proof}

In the proof of Lemma~\ref{co:csatisfiable}, we need to replace the
occurrence of~$\delta$ by~$\delta'$, so we can call
Lemma~\ref{le:lotofboundaries2} (instead of calling
Lemma~\ref{le:lotofboundaries}), and use Condition~C1' (instead of
using Condition~C1) to ensure that~$\delta' m > 0$, as is required in
the proof.  In the proof of Lemma~\ref{le:double}, we need to replace
the three occurrences of~$\delta$ by~$\delta'$, so we can call
Lemma~\ref{le:lotofboundaries2} (instead of calling
Lemma~\ref{le:lotofboundaries}), and use Condition~C4' (instead of
using Condition~C4) to ensure that~$\delta' c/p \geq k$,
as is required in the proof. Except for these changes, the structure
of the proof is exactly the same: Lemma~\ref{le:largev} states that
there is an instance~$\bI$ that is~$(\alpha,\beta)$-sparse such
that~$\bI \not\rightarrow \bS$, and Lemmas~\ref{co:csatisfiable}
and~\ref{le:double} give Lemma~\ref{lem:consistent};
i.e.,~$\bI \leq_k \bS$.

\section{Approximate Chromatic Number}
\label{sec:approximatecoloring}

For a graph~$\bG$ with~$n$ vertices, there is always an integer~$q$
satisfying~$1 \leq q \leq n$ such that~$\bG \to \bK_q$; the smallest
such~$q$ is the \emph{chromatic number} of~$\bG$. Any homomorphism
from~$\bG$ to~$\bK_q$ is called a proper~$q$-coloring of~$\bG$. The
problem of properly coloring a graph with as few as possible number
of colors has a long history. Finding the exact chromatic number is of
course one of the classical~NP-hard problems, straight from Karp's 21
list \cite{Karp1972}. The exact computational complexity of the
problem of approximating the chromatic number is much less understood,
despite the important progress on the problem since the discovery of
the PCP Theorem. In this section we study the width complexity
of the problem.

\subsection{Constant chromatic numbers}

In the regime of constant chromatic numbers, it is conjectured that
the problem of~$q$-coloring~$p$-colorable graphs is NP-hard for any
two constants~$p$ and~$q$ such that~$3 \leq p \leq q$. The larger the
gap between~$q$ and~$p$, the stronger the NP-hardness result.
For~$p = 3$, the current best result of this type, due to Barto, Bul\'{\i}n,
Krokhin, and Opr\v{s}al \cite{BBKO}, is that it is
NP-hard to~$5$-color~$3$-colorable graphs. For constant~$p \geq 4$,
the current best such result, due to Wrochna and \v{Z}ivn\'y
\cite{WrochnaZivny2020}, is that it is NP-hard
to~$q(p)$-color~$p$-colorable graphs,
where~$q(p) := \binom{p}{\lfloor{p/2}\rfloor}-1$. This improved over
the previously known point of NP-hardness for the
weaker~$q(p) = \exp(\Omega(p^{1/3}))$, which holds for sufficiently
large~$p$~\cite{Huang2013}, and for the even weaker~$q(p) = 2p-1$, which
holds for all~$p \geq 3$~\cite{BBKO}. The full
conjecture stating that the problem of~$q$-coloring~$p$-colorable
graphs is NP-hard for any two constants~$p$ and~$q$ such
that~$3 \leq p \leq q$ is known to follow from certain variants of the
\emph{Unique Games Conjecture}~(UGC)~\cite{DinurMR09}.

These results predict that the promise
problems~$\PCSP(\bK_{3},\bK_{5})$ and~$\PCSP(\bK_p,\bK_{q(p)})$
for~$p \geq 4$ are not solvable in bounded width, unless P = NP.  The
UGC-based results predict that~$\PCSP(\bK_p,\bK_q)$ is not solvable in
bounded width for any two constants~$p$ and~$q$ such
that~$3 \leq p \leq q$, unless either the suitable variants of the UGC
fail or P = NP. We show that our Main
Theorem~\ref{the:main} confirms all these predictions,
unconditionally, and in the stronger sense of ruling out, not only
solvability in constant width, but solvability in sublinear width:

\begin{theorem} \label{thm:charactcoloring}
For any two integers~$p$ and~$q$ such
  that~$1 \leq p \leq q$, the following statements are equivalent:
\begin{enumerate} \itemsep=0pt
\item[(a)] $\PCSP(\bK_p,\bK_q)$ is solvable by the consistency
  algorithm in width~$3$,
\item[(b)] $\PCSP(\bK_p,\bK_q)$ is solvable by the consistency
  algorithm in bounded width,
\item[(c)] $\PCSP(\bK_p,\bK_q)$ is solvable by the consistency
  algorithm in sublinear width,
\item[(d)] $p = 1$ or $p = 2$.
\end{enumerate} 
\end{theorem}

\begin{proof}
  If~$p = 1$, then the problem is trivial and solvable in width~2, and
  hence in width~3. If~$p = 2$, then the problem is solvable in
  width~$3$ since, in this case, the problem amounts to detecting
  whether~$\bG$ is 2-colorable, which is solvable in width~3 (see
  Example~\ref{ex:majority}). For~$p$ and~$q$ such
  that~$3 \leq p \leq q$, we
  have~$\graph(\bK_p) = E^{\bK_p} = \{ (i,j) \in [p]^2 \mid i \not= j
  \}$, and~$E^{\bK_p} \circ E^{\bK_p} = [p]^2$, while~$E^{\bK_q}$ is
  clearly irreflexive. Therefore, Theorem~\ref{the:main} implies
  that~$\PCSP(\bK_p,\bK_q)$ is not solvable in sublinear width, which
  completes the proof.
\end{proof}

\subsection{Algorithms with sublinear guarantees}

Turning to upper bounds, a well-known algorithm of
Wigderson~\cite{Wigderson1983} shows that the problem
of~$O(\sqrt{n})$-coloring 3-colorable graphs is solvable in polynomial
time.  We observe that, in its decision variant, Wigderson's algorithm
can be thought of as a \emph{width four} algorithm:

\begin{theorem}\label{th:ubapproxcoloring}
  For every graph~$\bG$, if~$\bG \leq_4 \clique_3$,
  then~$\bG \rightarrow \clique_{3\lceil{\sqrt{n}\rceil}}$, where~$n$
  is the number of vertices of~$\bG$.
\end{theorem}

\begin{proof}
  Fix a graph~$\bG$ with vertex-set~$V = [n]$.
  Let~$t = \lceil{\sqrt{n}}\rceil$. Assume
  that~$\bG \leq_4 \clique_3$. We prove, by induction on~$m$, that for
  every~$Y \subseteq V$ with~$|Y| \leq mt$, we
  have~$\bG|_Y \rightarrow \clique_{t+2m}$.
  Since~$t+2\lceil{n/t}\rceil \leq 3t$, the claim will follow by
  setting~$m = \lceil{n/t}\rceil$ and~$Y = V$.

  Fix~$m$ and~$Y \subseteq V$ with~$|Y| \leq mt$.  If~$m = 1$
  or~$\bG|_Y$ has all vertices of degree less than~$t$, then~$\bG|_Y$
  is~$t$-colorable,
  so~$\bG|_Y \rightarrow \clique_{t} \rightarrow \clique_{t+2m}$ and
  the claim is proved. Assume then that~$m > 1$ and~$\bG|_Y$ has some
  vertex~$v \in Y$ of degree at least~$t$.
  Let~$N \subseteq Y \setminus \{v\}$ be the neighborhood of~$v$
  in~$\bG|_Y$. Since~$\bG \leq_4 \clique_3$,
  also~$\bG|_{N \cup \{v\}} \leq_4 \clique_3$,
  implying~$\bG|_N \leq_3 \clique_2$. But then (see Example \ref{ex:majority}) $\bG|_N$
  is~$2$-colorable and hence~$\bG|_{N \cup \{v\}}$ is~$3$-colorable.
  Let~$h_v$ be a 3-coloring~$\bG|_{N \cup \{v\}}$ using the
  colors~$\{a,a+1,a+2\}$, where~$a := t+2m-2$ and~$h_v(v) = a$.
  Let~$X = Y \setminus (N \cup \{v\})$.
  Then~$|X| < |Y|-t \leq (m-1)t$ and hence, by induction
  hypothesis,~$\bG|_X \rightarrow \clique_{t+2(m-1)}$. Let~$g$ be a
  homomorphism from~$\bG|_X$ to~$\clique_{t+2m-2}$ and note
  that~$g \cup h_v$ is a homomorphism from~$\bG|_Y$
  to~$\clique_{t+2m}$; to see this, observe that the only possible
  color in~$\img(g) \cap \img(h_v)$ is~$a = h_v(v)$, but~$v$ is not
  adjacent in~$\bG|_Y$ to any vertex in~$\bG|_X$,
  since~$N \cap X = \emptyset$ by choice of~$X$. The proof is
  complete.
\end{proof}

In the rest of section we study the optimal width~$k = k(n)$ that
guarantees that, for any graph~$\bG$ with~$n$ vertices it holds
that~$\bG \leq_{k(n)} \bK_3$ implies~$\bG \to \bK_{O(n^\epsilon)}$ for
arbitrary but fixed~$\epsilon \in (0,1/2)$. 

As a first observation, it is straightforward to show
that~$k(n) \leq \lceil{n^{1-\epsilon}}\rceil$ suffices; i.e., for
every graph~$\bG$ with~$n$ vertices,
if~$\bG \leq_{\lceil{n^{1-\epsilon}}\rceil} \bK_3$,
then~$\bG \to \bK_{3\lceil{n^\epsilon}\rceil}$. Indeed, if~$V$ denotes
the set of vertices of~$\bG$, then one just selects a
subset~$X \subseteq V$ of~$\lceil{n^{1-\epsilon}}\rceil$ many
vertices, properly colors~$\bG|_X$ with three colors using the
assumption that~$\bG \leq_{\lceil{n^{1-\epsilon}}\rceil} \bK_3$, and
proceeds to the rest of the graph~$\bG|_{V \setminus X}$ with new
colors. Overall this uses at most~$3\lceil{n^\epsilon}\rceil$ colors.
We show that, by generalizing Wigderson's method, this width upper
bound of~$O(n^{1-\epsilon})$ can be improved
to~$O(n^{1-2\epsilon})$. This builds on some of the ideas of
Blum~\cite{Blum1994}, who showed that Wigderson's algorithm can be
improved to~$n^{3/8}$-color~$3$-colorable graphs in polynomial
time. Note that we do not achieve polynomial time, only sublinear
width, but we also show that this is necessary: for
any~$\gamma < 1-3\epsilon$, width~$O(n^{\gamma})$ does \emph{not}
suffice. Note the slight gap between the~$1-2\epsilon$ in the upper
bound and the~$1-3\epsilon$ in the lower bound.

\begin{theorem} \label{thm:genericbounds} Fix a
  real~$\epsilon \in (0,1/2)$ and an integer function~$q(n)$ such
  that~$q(n) \geq 3$ holds for all integers~$n \geq 1$,
  and~$q(n) = \Theta(n^\epsilon)$. The following statements hold:
\begin{enumerate}
\item There is an integer function~$k(n) = O(n^{1-2\epsilon})$
  such that, for every graph~$\bG$, if~$\bG \leq_{k(n)} \bK_3$
  then~$\bG \to \bK_{q(n)}$, where~$n$ is the number of vertices
  of~$\bG$.
\item For every real~$\gamma < 1-3\epsilon$ and every integer
  function~$k(n) = O(n^{\gamma})$ there exist
  arbitrarily large graphs~$\bG$ such that~$\bG \leq_{k(n)} \bK_3$
  and~$\bG \not\to \bK_{q(n)}$, where~$n$ is the number of vertices
  of~$\bG$.
\end{enumerate}
\end{theorem}

\begin{proof}[Proof of~1] 
  Choose~$k(n) := \max\{3+\lceil{C^2 n^{1-2\epsilon}}\rceil,n_0\}$ for
  sufficiently large integers~$C$ and~$n_0$ to be determined later;
  their choice will depend on the constants that are implicit in
  the assumption that~$q(n) = \Theta(n^\epsilon)$. Note that~$k(n)$ is
  indeed~$O(n^{1-2\epsilon})$. Fix any integer~$n$ and let~$\bG$ be a
  graph with vertex-set~$V$ and edge-set~$E$ with~$|V|=n$. Assume
  that~$\bG \leq_{k(n)} \clique_3$ and let~$\script{H}$ be
  a~$k(n)$-strategy on~$\bG$ and~$\clique_3$.  We shall prove that for
  every~$Y\subseteq V$ there exists a set~$X\subseteq Y$ such
  that~$\bG|_X$ is~$3$-colorable
  and~$|X|\geq \min\{m,Cm^{1-\epsilon}\}$ where~$m=|Y|$.
  
  Let~$Y\subseteq V$. For every~$S\subseteq Y$ we use~$N(S)$ to denote
  the neighbourgood of~$S$ in the subgraph~$\bG|_Y$ induced by~$Y$
  in~$\bG$, i.e., the set of vertices in~$Y$ that are adjacent
  in~$\bG|_Y$ to some vertex in~$S$.  If~$|Y| \leq k(n)$, then, by the
  extension property up to~$k(n)$, the strategy~$\script{H}$ contains
  a mapping~$h$ with domain~$Y$. By definition,~$h$ is a partial
  homomorphism from~$\bG$ to~$\clique_3$, and hence a
  proper~$3$-coloring of~$\bG|_Y$. Thus, we can just set~$X=Y$ and the
  claim is proved
  since~$|X| = |Y| = m \geq \min\{m,Cm^{1-\epsilon}\}$.  Assume then
  that~$|Y| > k(n)$. Now, consider two cases: (a) there exists a
  subset~$S\subseteq Y$ with~$|S|=k(n)-3$ such
  that~$|S\cup N(S)|>Cm^{1-\epsilon}$, and (b) such a set does not
  exist.

  \emph{Case (a)}. We shall prove that in this case~$\bG|_X$
  is~$3$-colorable for~$X = S \cup N(S)$, which proves the claim
  since~$|X| > Cm^{1-\epsilon}$. First note that, by the extension
  property up to~$k(n)$, the strategy~$\script{H}$ contains a
  mapping~$h:S\rightarrow [3]$ with domain~$S$; this follows from the
  fact that~$|S|=k(n)-3 \leq k(n)$. Furthermore,~$h$ is a partial
  homomorphism, so it is a proper~$3$-coloring of~$\bG|_S$. We shall
  prove that~$h$ can be extended to a proper~$3$-coloring of~$\bG|_X$.
  For every node~$v\in N(S)$, let~$L_v\subseteq [3]$ be the set of
  {\em remaining} colors for~$v$, i.e., the
  set~$[3] \setminus\{h(u) \mid u \in S \cap N(\{v\})\}$. Note
  that~$|L_v|\leq 2$ for each~$v \in N(S)$. We want to show that there
  exists a {\em list} coloring of~$\bG|_{N(s)}$, i.e., a coloring~$g$
  of~$\bG|_{N(s)}$ such that~$g(v)\in L_v$ for every~$v\in
  N(S)$. Consider the CSP instance~$(\bA,\bB)$ where~$A=N(S)$
  and~$B=[3]$, whose signature~$\sigma$ contains a relation~$R_{u,v}$
  for every pair~$u,v$ of different elements in~$N(S)$. In particular,
  for every~$u,v \in N(S)$ with~$u \not=v$ we
  have~$R^{\bA}_{u,v}=\{(u,v)\}$
  and~$R^{\bB}_{u,v}=[3]^2\cap (L_u\times L_v)$. It follows
  immediately from the definition of~$\bA$ and~$\bB$ that for every
  mapping~$g:A\rightarrow B$, it holds that~$g$ is a homomomorphism
  from~$\bA$ to~$\bB$ if and only if~$g$ is a list-coloring
  of~$\bG|_{N(S)}$.  Since $|L_v|\leq 2$ for all~$v\in A$ then
  $\Pol(\bB)$ contains the function $\varphi:B^3\rightarrow B$ defined
  as
$$
\varphi(x,y,z)=\left\{
\begin{array}{ll}
x & \text{if } x=y \\
z & \text{otherwise}
\end{array}\right.
$$
This is a majority operation (see Example~\ref{ex:majority}) and,
therefore, it follows that~$\bB$ has width~$3$. Consider the
set~$\script{H'}$ of partial maps from~$\bA$ to~$\bB$ that
contains~$f|_{X\setminus S}$ for every~$f \in \script{H}$ with
domain~$S$ and every set~$X\subseteq V$ such that~$S\subseteq X$
and~$|X|\leq k(n)$. It is easy to see that~$\mathcal H'$ is
a~$3$-strategy on~$\bA$ and~$\bB$, hence~$\bA \leq_3 \bB$. Since~$\bB$
has width~$3$ it follows that~$\bA \to \bB$ and we are done.

\emph{Case (b)}. Let~$\script{S}$ be any maximal collection of subsets
of~$Y$ of cardinality~$k(n)-3$ such
that~$(S\cup N(S))\cap (S'\cup N(S'))=\emptyset$ for any two
distinct~$S,S'\in \script{S}$. Note that, as in Case (a), we have
that~$\bG|_S$ has a proper~$3$-coloring for every~$S\in
\script{S}$. Then, if we let~$X=\bigcup_{S\in \script{S}} S$,
then~$\bG|_X$ also has a proper~$3$-coloring.
Since~$|S\cup N(S)|\leq Cm^{1-\epsilon}$ for every~$S\in \script{S}$
and~$|Y|=m$, we have~$|\script{S}|\geq m^{\epsilon}/C$. It then
follows that~$X$ has cardinality at
least~$(m^{\epsilon}/C)\cdot (k(n)-3)\geq Cm^{1-\epsilon}$, by the
choice of~$k(n)$ and the fact that~$n \geq m$. This finishes the proof
of the claim.

The rest of the proof is fairly standard. The following recursive
algorithm produces a valid coloring for~$\bG|_Y$ for
any~$Y\subseteq V$. Let~$m := |Y|$.  If~$m \leq Cm^{1-\epsilon}$, then
color~$\bG|_Y$ with three colors, which is possible since in this
case~$\min\{m,Cm^{1-\epsilon}\} = m = |Y|$. Else, select
an~$X \subseteq Y$ with~$|X| \geq Cm^{1-\epsilon}$ such that~$G|_X$
has a proper~$3$-coloring~$g$, which is again possible since in this
case~$\min\{m,Cm^{1-\epsilon}\} = Cm^{1-\epsilon}$.  Recursively,
obtain a proper coloring~$h$ of~$\bG|_{Y\setminus X}$. Renaming colors
if necessary we can assume that~$g$ and~$h$ do not use any color in
common. Return~$g\cup h$. Note that~$g\cup h$ uses at most~$Q(m)$
different colors where~$Q(m)$ is the solution to the following
recurrence:
\begin{equation*}
\begin{array}{lll}
Q(m)=3+Q(m-\lceil Cm^{1-\epsilon}\rceil) & & \text{if } m > Cm^{1-\epsilon} \\
Q(m)=3 & & \text{otherwise}
\end{array}
\end{equation*}

We shall show that the bound~$Q(n) \leq q(n)$ holds for
every~$n\geq n_0$ for sufficiently large~$n_0$. The statement will
follow since the choice of~$k(n)$ guarantees~$k(n) \geq n_0$ and,
therefore, any graph~$\bG$ with~$n < n_0$ many vertices that
satisfies~$\bG \leq_{k(n)} \bK_3$ is even 3-colorable. Recall
that~$q(n) \geq 3$ holds by assumption.

To prove our claim, we first note that, whenever~$m \geq n/2$, we
have~$\lceil Cm^{1-\epsilon}\rceil\geq C(n/2)^{1-\epsilon}$, and
therefore it takes at
most~$\lceil{(n/2)/(C(n/2)^{1-\epsilon})}\rceil =
\lceil{(n/2)^{\epsilon}/C}\rceil$ many iterations of the recurrence to
get from~$Q(n)$ down to~$Q(\lfloor{n/2}\rfloor)$. It follows that
\begin{equation}
  Q(n)\leq 3 \lceil{(n/2)^{\epsilon}/C}\rceil + Q(\lfloor n/2\rfloor)
  \leq 3 (n/2)^{\epsilon}/C + 3 +  Q(\lfloor n/2\rfloor). \label{eqn:iteite}
\end{equation}
Iterating this recurrence we get
\begin{equation}
Q(n) \leq (3/C) n^\epsilon \textstyle{\sum_{i \geq 1} (1/2^{\epsilon})^{i}} + 3\log_2(n) 
\leq  (9/C) n^{\epsilon} + 3\log_2(n) \leq q(n),
\end{equation}
where the second inequality follows from the
identity~$\sum_{i \geq 1} r^i = r/(1-r)$ and the fact
that~$1/2^\epsilon \leq 1/\sqrt{2}$ holds since~$\epsilon \leq
1/2$. The third inequality follows from the assumption
that~$q(n) = \Theta(n^\epsilon)$ and~$n \geq n_0$, provided the
constants~$C$ and~$n_0$ are chosen large enough.
\end{proof}

\begin{proof}[Proof of 2.]
  Fix~$\gamma < 1-3\epsilon$ and~$k(n) = O(n^{\gamma})$. We analyze
  the probabilistic construction in the proof of
  Theorem~\ref{the:main} when~$\bS = \bK_p$ and~$\bT = \bK_q$
  for~$p = 3$ and~$q = q(n) = \Theta(n^\epsilon)$. Note that~$q$ is
  now also a function of~$n$ and that the argument presented in
  Section~\ref{sec:proof} allows this generality. Since~$r = 2$, we
  use the version of the proof in
  Section~\ref{sec:digraphs}. Furthermore, as in the proof of
  Theorem~\ref{thm:charactcoloring}, the left template~$\bS = \bK_p$
  satisfies Assumption~A1, and it is obvious that the right
  template~$\bT = \bK_q$ satisfies Assumption~A2. Thus, we just need
  to set the parameters in the proof.

  The given data is~$r,p,q,k$,
  where~$q = q(n)$ and~$k = k(n)$ are functions of~$n$, and we need to
  produce, for every large enough integer~$n$, a choice of the real
  parameters~$\delta,\beta,\alpha,c,d$, that may or may not be
  functions of~$n$, in such a way that
  Conditions~C1',C2,C3,C4',C5,C6,C7 hold, where Conditions~C1' and~C4'
  are stated in Section~\ref{sec:digraphs}, and
  Conditions~C2,C3,C5,C6,C7 are stated in
  Section~\ref{sec:stage}. Once we achieve this, Lemma~\ref{le:largev}
  will provide a graph~$\bG$ with~$n$ vertices that
  is~$(\alpha,\beta)$-sparse such that~$\bG \not\rightarrow \bK_q$,
  holds. By Lemma~\ref{lem:consistent} (derived as in
  Section~\ref{sec:digraphs}), this~$\bG$ will also
  satisfy~$\bG \leq_{k} \bK_p$, and since this will succeed for any
  large enough~$n$, the statement will be proved.

  Set~$\delta_0 := \epsilon/(1-\epsilon-\gamma)$ and note that the
  assumption~$\gamma < 1-3\epsilon$ implies~$\delta_0 <
  1/2$. Set~$\delta$ to be any positive real in the
  interval~$(\delta_0,1/2)$. The upper bound~$\delta < 1/2$ means that
  Condition~C1' holds. For later use, we note that the lower
  bound~$\delta_0 < \delta$ implies
  \begin{equation}
  \gamma < 1-(1+\delta)\epsilon/\delta. \label{eqn:deltachoice}
\end{equation}
Set~$\beta := 1+\delta$, so Condition~C2 is satisfied.
Set~$d = d(n) := 5q(n)\ln(q(n))$, so Condition~C6 is satisfied for all
large enough~$n$.
Set~$\alpha = \alpha(n) := (C/d(n))^{(1+\delta)/\delta}$
where~$C := (1+\delta) 4^{\delta/(1+\delta)}
e^{(-4-3\delta)/(1+\delta)}$. Observe that~$d(n)$ is an increasing
function of~$n$, while~$C,\beta,r$ are constants independent
of~$n$. In particular the rate at which~$\alpha(n)$ approaches~$0$ is
that of~$(1/d(n))^{(1+\delta)/\delta}$, and the rate at which the
right-hand side in Condition~C3 approaches~$0$ is that
of~$(1/d(n))^{1/(r-1)}$.  Since~$1/(r-1) = 1 < (1+\delta)/\delta$,
this means that Condition~C3 holds for all large enough~$n$. Finally,
set~$c = c(n) := (p/\delta')k(n)$
for~$\delta' := (1-2\delta)/(6(1+\delta))$, so Condition~C4' is
satisfied for all~$n$. We need to argue that Conditions~C5 and~C7 hold
for all large enough~$n$

To argue that Condition~C5 holds we need to show
that~$n \geq \max\{c(n)/(\alpha(n)\beta),q(n)\}$ for all large
enough~$n$.  Clearly~$n \geq q(n)$ for all large enough~$n$
since~$q(n) = \Theta(n^\epsilon)$ and~$\epsilon < 1/2$. Thus, it
suffices to show that~$n \geq c(n)/(\alpha(n)\beta)$ or, equivalently,
that~$k(n) \leq f(n) := (\beta\delta'/p)\alpha(n)n$, for all large
enough~$n$. First, recall that the
parameters~$\epsilon,p,\beta,C,\delta'$ are constants independent
of~$n$. Therefore, recalling
that~$d(n) = 5q(n)\ln(q(n)) = \Theta(\epsilon n^\epsilon \ln(n))$, the
growth rate of~$f(n)$ is that
of~$n^{1-(1+\delta)\epsilon/\delta} \ln(n)^{-(1+\delta)/\delta}$.  On
the other hand, the growth rate of~$k(n)$ is bounded above by that
of~$n^{\gamma}$. Now, the choice of~$\delta$
guarantees~\eqref{eqn:deltachoice} and, therefore, we
have~$k(n) = o(f(n))$. It follows that, for all large enough~$n$, it
holds that~$k(n) \leq f(n)$, which means that Condition~C5 holds.

  To argue that Condition~C7 holds, observe that the choice of~$d(n)$
  ensures that~$p_1(n) \rightarrow 0$ as~$n \rightarrow \infty$, and
  the choice of~$\alpha(n)$ (and the constant~$C$) ensures
  that~$p_2(n)$ is bounded by~$\sum_{v \geq 1} (1/4)^v = 1/3$ for
  all~$n$. Both facts together imply that~$p_1(n) + p_2(n) < 1$ and
  Condition~C7 holds for all large enough~$n$.
\end{proof}

Recall from Section~\ref{sec:consistency} that there is an algorithm
that, given a graph~$\bG$ and an integer~$k$, decides
whether~$\bG \leq_k \bK_3$ in time polynomial in~$n^k$, where~$n$ is
the number of vertices of~$\bG$ and, if so, returns a
strategy~$\script{H}$. The proof of the width upper bound in
Theorem~\ref{thm:genericbounds} gives the following:

  \begin{theorem} \label{thm:algorithm} Fix a
    real~$\epsilon \in (0,1/2)$ and an integer function~$q(n)$ such
    that~$q(n) \geq 3$ holds for all integers~$n \geq 1$,
    and~$q(n) = \Theta(n^\epsilon)$. Then, there is an algorithm that
    finds a proper~$q(n)$-coloring of any given~$3$-colorable graph
    with~$n$ vertices in~$2^{\Theta(n^{1-2\epsilon}\log(n))}$-time.
\end{theorem}

\begin{proof}
  We analyse the recursive algorithm given in the proof of the first
  part of Theorem~\ref{thm:genericbounds}.  The algorithm starts by
  computing a strategy~$\script{H}$ that
  witnesses~$\bG \leq_{k(n)} \bK_3$; such a strategy exists because,
  indeed, the assumption is that~$\bG \to \bK_3$. The runtime of this
  step is polynomial in~$n^{k(n)}$. Once~$\script{H}$ is computed, the
  algorithm proceeds recursively as described in the proof of
  Theorem~\ref{thm:genericbounds} starting at~$Y = V$, where~$V$ is
  the set of all vertices of~$\bG$.  To find the required
  set~$X \subseteq Y$ with~$|X| \geq \min\{m,Cm^{1-\epsilon}\}$ for
  the~$Y$ of cardinality~$m$ in the current recursive call, we first
  need to tell whether~$Y$ falls in Case (a) or in Case (b). For this,
  it suffices to loop through all subsets~$S \subseteq Y$
  with~$|S| = k(n)-3$, and compute~$|S \cup N(S)|$. The number of such
  sets is bounded by~$n^{k(n)-3}$ and hence can be looped in time
  polynomial in~$n^{k(n)}$.  In Case (a), we color~$\bG|_X$
  for~$X = S \cup N(S)$ with~$3$ colors as follows: first find
  an~$h \in \script{H}$ with~$\dom(h) = S$, and then extend~$h$ to a
  proper 3-coloring of~$\bG|_X$ by solving the CSP
  instance~$(\bA,\bB)$.  We are using here the well-known fact (see
  \cite{Cohen2004}) that for CSPs an polynomial-time algorithm for the
  decision variant yields immediately and polynomial-time algorithm
  for the search version.  In Case (b), we greedily find a maximally
  disjoint family~$\script{S}$ of sets of the form~$S \cup N(S)$
  with~$S \subseteq Y$, and color~$\bG|_X$
  for~$X = \bigcup_{S \in \script{S}}$ with three colors
  as~$h = \bigcup_{S \in \script{S}} h_S$, where~$h_S \in \script{H}$
  with~$\dom(h_S) = S$ is a suitably found proper 3-coloring
  of~$\bG|_S$ for each~$S \in \script{S}$.  Each recursive call
  shrinks the size of the calling set~$Y$ from~$m$
  to~$m-\lceil{Cm^{1-\epsilon}}\rceil$, which means that the algorithm
  ends after a linear in~$n$ number of recursive calls, each of which
  takes time polynomial in~$n^{k(n)}$.
  For~$k(n) = \Theta(n^{1-2\epsilon})$, this is time
  complexity~$2^{\Theta(n^{1-2\epsilon}\log(n))}$ overall, and the
  proof is complete.
\end{proof}

\subsection{Discussion and an open problem}

Some discussion on the runtime of the algorithm in
Theorem~\ref{thm:algorithm} is in order. On one hand, the simple
observation we made just before the statement of
Theorem~\ref{thm:algorithm} that width~$\lceil{n^{1-\epsilon}}\rceil$
suffices already gives a very simple algorithm that
properly~$\Theta(n^\epsilon)$-colors~$3$-colorable graphs with~$n$
vertices in subexponential~$2^{\Theta(n^{1-\epsilon})}$-time. The
algorithm of Theorem~\ref{thm:algorithm} is only slightly more
complicated and asymptotically beats this. On the other hand, using
more sophisticated techniques, it was shown in \cite{BansalCLNN19}
that, for any desired approximation factor~$f$, there is
a~$2^{\tilde\Theta(n/(f\log(f) + f\log(f)^2))}$-time randomized
algorithm that approximates the chromatic number of a graph with~$n$
vertices within a factor of~$f$.  For~$f = \Theta(n^\epsilon)$, this
gives a~$2^{\Theta(n^{1-\epsilon-o(1)})}$-time randomized algorithm
for the problem of~$\Theta(n^\epsilon)$-coloring~$3$-colorable graphs
with~$n$ vertices. Interestingly, the simple width-based algorithm
from Theorem~\ref{thm:algorithm} also beats this, and is deterministic
(but of course it applies only to our problem and not to the more
general problem of approximating the chromatic number).

Whether the runtime~$2^{\Theta(n^{1-2\epsilon}\log(n))}$ of
Theorem~\ref{thm:algorithm} can be beaten is an interesting question
left open by our work. Our width lower bound of~$n^{1-3\epsilon}$ has
as a consequence that~$2^{\Omega(n^{1-3\epsilon})}$ appears to be a
lower limit on the runtime of any width-based algorithm. The obstacle
to improving the width lower bound from~$n^{1-3\epsilon}$
to~$n^{1-2\epsilon}$ is Condition~C1', which is an improvement for the
special case of graphs over Condition~C1 of the general case. Ideally,
Condition C1' should be improved further to Condition~C1'' defined
as~$0 < \delta < 1$. We do not know if this is possible; we leave it
as an open problem.

\end{document}